\RequirePackage{fix-cm}
\documentclass[numbook, envcountsect, envcountsame, envcountreset, runningheads, smallextended]{svjour3}

\usepackage{amsmath}
\usepackage{amssymb}
\usepackage{xcolor}
\usepackage{verbatim}
\usepackage{tikz}
\usepackage{pstricks}
\usepackage{pgfplots}
\usepackage{mathptmx}
\usepackage[shortcuts]{extdash}
\DeclareMathOperator*{\esssup}{ess\,sup}

\newtheorem{thm}{Theorem}[section]
\newtheorem{defn}[thm]{Definition}
\newtheorem{rem}[thm]{Remark}
\newtheorem{lem}[thm]{Lemma}
\newtheorem{prop}[thm]{Proposition}
\newtheorem{coro}[thm]{Corollary}

\newtheorem{eg}[thm]{Example}

\renewcommand{\P}{\mathbb{P}}
\newcommand{\F}{\mathbb{F}}

\newcommand{\Q}{\mathbb{Q}}
\newcommand{\R}{\mathbb{R}}
\newcommand{\E}{\mathbb{E}}
\newcommand{\N}{\mathbb{N}}
\newcommand{\cF}{\mathcal{F}}
\newcommand{\cH}{\mathcal{H}}

\newcommand{\cU}{\mathcal{U}}
\newcommand{\cP}{\mathcal{P}}
\newcommand{\cQ}{\mathcal{Q}}

\newcommand{\cE}{\mathcal{E}}
\newcommand{\cM}{\mathcal{M}}
\newcommand{\A}{\mathcal{A}}

\newcommand{\cS}{\mathcal{S}}

\newcommand{\eps}{\varepsilon}
\def\bcdot{\text{ \tiny\textbullet }~ }

\numberwithin{equation}{section}

\title{Model-independent superhedging under portfolio constraints}
\author{Arash Fahim \and Yu-Jui Huang}
\institute{A. Fahim \at Department of Mathematics, Florida State University\\
\email{fahim@math.fsu.edu}
\and
Y.-J. Huang \at School of Mathematical Sciences, Dublin City University \\ 
\email{yujui.huang@dcu.ie}\\
{We thank Pierre Henry-Labord{\`e}re, Constantinos Kardaras, and Jan Ob{\l}{\'o}j for their thoughtful suggestions. We are also thankful to the anonymous referees for their elaborate comments which contribute to the quality of this work.}\\
{A. Fahim is partially supported by Florida State University CRC FYAP (315-81000-2424) and the NSF (DMS-1209519).}\\
{Y.-J. Huang is partially supported by SFI (07/MI/008 and 08/SRC/FMC1389) and the ERC (278295).}
}

\begin{document}
\maketitle

\begin{abstract}
In a discrete-time market, we study model-independent superhedging, while the semi-static superhedging portfolio consists of {\it three} parts: static positions in liquidly traded vanilla calls, static positions in other tradable, yet possibly less liquid, exotic options, and a dynamic trading strategy in risky assets under certain constraints. By considering the limit order book of each tradable exotic option and employing the Monge-Kantorovich theory of optimal transport, we establish a general superhedging duality, which admits a natural connection to convex risk measures. With the aid of this duality, we derive a model-independent version of the fundamental theorem of asset pricing. The notion ``finite optimal arbitrage profit'', weaker than no-arbitrage, is also introduced. It is worth noting that our method covers a large class of Delta constraints as well as Gamma constraint.

\end{abstract}

\section{Introduction}
To avoid model mis-specification, one may choose to consider only ``must-be-true'' implications from the market. The standard approach, suggested by Dupire \cite{Dupire94}, leverages on market prices of liquidly traded vanilla call options: one does not manage to specify a proper physical measure, but considers all measures that are consistent with market prices of vanilla calls as plausible pricing measures. These measures then provide model-independent bounds for prices of illiquid exotic options, and motivate the practically useful semi-static hedging, which involves static holdings in vanilla calls and dynamic trading in risky assets. Pioneered by Hobson \cite{Hobson98}, this thread of research has drawn substantial attention; see e.g. \cite{BHR01}, \cite{BP02}, \cite{HLW05-QF}, \cite{HLW05-IME}, \cite{LW05}, \cite{CDDV08}, \cite{CO11}, \cite{BHP13}, and \cite{DS14-PTRF}. In particular, Beiglb{\"o}ck, Henry-Labord{\`e}re \& Penkner establish in \cite{BHP13} a general duality of model-independent superhedging, under a discrete-time setting where market prices of vanilla calls with maturities {\it at or before} the terminal time $T>0$ are all considered.

In reality, what we can rely on goes beyond vanilla calls. In the markets of commodities, for instance, Asian options and calendar spread options are largely traded, with their market or broker quotes easily accessible. In the New York Stock Exchange and the Chicago Board Options Exchange, standardized digital and barrier options have been introduced, mostly for equity indexes and Exchange Traded Funds.
What we can take advantage of, as a result, includes market prices of not only vanilla calls, but also certain tradable exotic options.

In this paper, we take up the model-independent framework in \cite{BHP13}, and intend to establish a general superhedging duality, under the consideration of additional tradable options besides vanilla calls, as well as portfolio constraints on trading strategies in risky assets. More specifically, our semi-static superhedging portfolio consists of three parts: 1. static positions in liquidly traded vanilla calls, as in the literature of robust hedging; 2. static positions in additional tradable, yet possibly less liquid, exotic options; 3. a dynamic trading strategy in risky assets under certain constraints. 

While tradable, the additional exotic options may be very different from vanilla calls, in terms of liquidity. Their limit order books are usually very shallow and admit large bid-ask spreads, compared to those of the underlying assets and the associated vanilla calls. It follows that we need to take into account the whole limit order book, instead of one single market quote, of each of the additional options, in order to make trading possible. We formulate the limit order books in Section~\ref{sec:options}, and consider the corresponding non-constant unit price functions. On the other hand, portfolio constraints on trading strategies in risky assets have been widely studied under the model-specific case; see \cite{CK93} and \cite{JK95} for deterministic convex constraints, and \cite{FK97}, \cite{CPT01}, \cite{Napp03}, and \cite{Rokhlin05}, among others, for random and other more general constraints. Our goal is to place portfolio constraints under current model-independent context, and investigate its implication to semi-static superhedging. 
 
We particularly consider a general class of constraints which enjoys adapted convexity and continuous approximation property (Definition~\ref{defn:cS}). This already covers a large collection of Delta constraints, including adapted convex constraints; see Remark~\ref{rem:adapted convex}. For the simpler case where no additional tradable option exists, we derive a superhedging duality in Proposition~\ref{thm:duality}, by using the theory of optimal transport. This in particular generalizes the duality in \cite{BHP13} to the multi-dimensional case with portfolio constraints; see Remarks~\ref{rem:S=H} and \ref{rem:multi-d}. Then, on strength of the convexity of the non-constant unit price functions, we are able to extend the above duality to the general case where additional tradable options exist; see Theorem~\ref{thm:duality main}. Note that Acciaio, Beiglb{\"o}ck, Penkner \& Schachermayer \cite{ABPS13} also applies to model-independent superhedging in the presence of tradable exotic options, while assuming implicitly that each option can be traded liquidly. Theorem~\ref{thm:duality main} can therefore be seen as a generalization of \cite{ABPS13} that deals with different levels of liquidity; see Remark~\ref{rem:relate to [1]}.

The second part of the paper investigates the relation between the superhedging duality and the fundamental theorem of asset pricing (FTAP). It is well known in the classical model-specific case that the FTAP yields the superhedging duality. This relation has been carried over to the model-independent case by \cite{ABPS13}, where an appropriate notion of model-independent arbitrage was introduced. In the same spirit as in \cite{ABPS13}, we define model-independent arbitrage in Definition~\ref{defn:NA}, under current setting with additional tradable options and portfolio constraints. With the aid of the superhedging duality in Theorem~\ref{thm:duality main}, we are able to derive a model-independent FTAP; see Theorem~\ref{thm:NA}. While the theorem itself does not distinguish between arbitrage due to risky assets and arbitrage due to additional tradable options, Lemmas~\ref{lem:P_S characterization} and \ref{lem:cE^Q_I=0} can be used to differentiate one from the other. It is also worth noting that we derive the FTAP as a consequence of the superhedging duality. This argument was first observed in \cite{DS14}, as opposed to the standard argument of deriving the superhedging duality as a consequence of the FTAP, used in both the model-specific case and \cite{ABPS13}.

With the FTAP (Theorem~\ref{thm:NA}) at hand, we observe from Theorem~\ref{thm:duality main} and Proposition~\ref{prop:risk measure} that the problems of superhedging and risk-measuring can be well-defined even when there is model-independent arbitrage to some extent. We relate this to optimal arbitrage under the formulation of \cite{CT13}, and show that superhedging and risk-measuring are well-defined as long as ``the optimal arbitrage profit is finite'', a notion weaker than no-arbitrage; see Proposition~\ref{thm:NUP}. We also compare Theorem~\ref{thm:NA} with \cite[Theorem 9.9]{FS-book-11}, the classical model-specific FTAP under portfolio constraints, and observe that a closedness condition in \cite{FS-book-11} is no longer needed under current setting. An example given in Section~\ref{sec:compare classical} indicates that availability of vanilla calls obviates the need of the closedness condition.

Finally, we extend our scope to Gamma constraint. While Gamma constraint does not satisfy adapted convexity in Definition~\ref{defn:cS} (ii), it admits additional boundedness property. Taking advantage of this, we are able to modify previous results to obtain the  corresponding superhedging duality and FTAP in Propositions~\ref{prop:duality bdd} and \ref{prop:NA bdd}.

This paper is organized as follows. In Section \ref{sec:setup}, we prescribe the set-up of our studies. In Section \ref{sec:duality}, we establish the superhedging duality, and investigate its connection to other dualities in the literature and convex risk measures. In Section \ref{sec:NA}, we define model-independent arbitrage under portfolio constraints with additional tradable options, and derive the associated FTAP. The notion ``finite optimal arbitrage profit'', weaker than no-arbitrage, is also introduced. Section \ref{sec:examples} presents concrete examples of portfolio constraints and the effect of additional tradable options. Section \ref{sec:bounded} deals with constraints which do not enjoy adapted convexity, but admit some boundedness property. 
Appendix~\ref{sec:appendix} contains a counter-example which emphasizes the necessity of the continuous approximation property required in Definition~\ref{defn:cS}.


\section{The set-up}\label{sec:setup}
We consider a discrete-time market, with a finite horizon $T\in\N$. There are $d$ risky assets $S=\{S_t\}_{t=0}^T=\{(S_t^1,\dots,S_t^d)\}_{t=0}^T$, whose initial price $S_0=x_0\in\R_+^d$ is given. There is also a risk-free asset $B=\{B_t\}_{t=0}^T$ which is normalized to $B_t\equiv 1$. Specifically, we take $S$ as the canonical process $S_t(x_1,x_2,\dots,x_T)=x_t$ on the path-space ${\Omega:=(\R^d_+)^T}$, and denote by $\F=\{\cF_t\}_{t=0}^T$ the natural filtration generated by $S$. 

\subsection{Vanilla calls and other tradable options}\label{sec:options}
At time $0$, we assume that the vanilla call option with payoff $(S^n_t-K)^+$ can be liquidly traded, at some price $C_n(t,K)$ given in the market, for all $n=1,\dots,d$, $t=1,\dots,T$, and $K\ge 0$. The collection of pricing measures consistent with market prices of vanilla calls is therefore 
\begin{equation}\label{defn Q}
\Pi\!\!:=\!\!\bigl\{\Q\in\cP(\Omega):\E^\Q[(S^n_t-K)^+]=C_n(t,K), \forall n=1,\dots,d, t=1,\dots,T,\hbox{and } K\ge 0\bigr\},
\end{equation}
where $\cP(\Omega)$ denotes the collection of all probability measures defined on $\Omega$.

In view of \cite[Proposition 2.1]{HR12}, for each $n=1,\dots,d$ and $t=1,\dots,T$, as long as $K\mapsto C_n(t,K)$ is nonnegative, convex and satisfies $\lim_{K\downarrow 0}\partial_KC_n(t,K)\ge-1$, and $\lim_{K\to\infty} C_n(t,K) = 0$, the relation $\E^\Q[(S^n_t-K)^+]=C_n(t,K)$ $\forall K\ge0$ already prescribes the distribution of $S^n_t$ on $\R_+$, which will be denoted by $\mu^n_t$. Thus, by setting $\Q^n_t$ as the law of $S^n_t$ under $\Q$, we have
\begin{equation}\label{Pi 2}
\Pi=\left\{\Q\in\cP(\Omega): \Q^n_t =\mu^n_t,\ \forall\  n=1,\dots,d\ \hbox{and}\ t=1,\dots,T\right\}.
\end{equation}

\begin{rem}\label{rem:E[S] finite}
Given $\Q\in\Pi$, note that $\E^\Q[S^n_t]<\infty$ for all $n=1,\dots,d$ and ${t=1,\dots,T}$ (which can be seen by taking $K=0$ in \eqref{defn Q}).
\end{rem}

\begin{rem}\label{rem:Pi compact}
In view of \eqref{Pi 2}, $\Pi$ is nonempty, convex, and weakly compact. This is a direct consequence of \cite[Proposition 1.2]{Kellerer84}, once we view $\Omega=(\R^d_+)^T$ as the product of $(d\times T)$ copies of $\R_+$.
\end{rem}

\begin{rem}
We do not assume that $t\mapsto C_n(t,K)$ is increasing for each fixed $n$ and $K$. This condition, normally required in the literature (see e.g. \cite[p. 481]{BHP13}), implies that the set of martingale measures
\begin{equation}\label{def:cM}
\cM:=\{\Q\in\Pi:S=\{S_t\}_{t=0}^T\ \hbox{is a martingale under}\ \Q\}
\end{equation}
is non-empty, which underlies the superhedging duality in \cite{BHP13}. In contrast, the superhedging duality in Proposition~\ref{thm:duality} below hinges on a different collection $\cQ_\cS$ which contains $\cM$ (see Definition~\ref{def:Q_S}). Since it is possible that our duality holds while $\cM=\emptyset$, imposing ``$t\mapsto C_n(t,K)$ is increasing'' is not necessary.    
\end{rem}

Besides vanilla calls, there are other options tradable, while less liquid, at time $0$. Let $I$ be a (possibly uncountable) index set. For each $i\in I$, suppose that $\psi_i:\Omega\mapsto\R$ is the payoff function of an option tradable at time $0$. Let $\eta\in\R$ be the number of units of $\psi_i$ being traded at time $0$, with $\eta\ge 0$ denoting a purchase order and $\eta<0$ a selling order. Let $c_i(\eta)\in\bar\R:=\R\cup\{-\infty,+\infty\}$ denote the total cost of trading $\eta$ units of $\psi_i$.  Throughout this paper, we impose the following condition:
\begin{equation}\tag{C}
\hbox{for all}\ i\in I, \hbox{the map}\ \eta\mapsto c_i(\eta)\ \hbox{is convex with $c_i(0)=0$}.
\end{equation}
 We can then define the unit price $p_i(\eta)$ for trading $\eta$ units of $\psi_i$ by 
\[
p_i(\eta):=\frac{c_i(\eta)}{\eta}\ \quad \hbox{for}\ \eta\in\R\setminus\{0\},\quad\hbox{and}\quad  p_i(0):=c'_i(0+).
\] 

\begin{rem}
Condition (C) is motivated by the typical structure of a limit order book of a nonnegative option, as demonstrated in Figure~\ref{fig:bid-ask-chart}. That is, the option $\psi_i$ can be purchased only at prices $0\le a_1 \le a_2\le\dots \le a_{\ell}$ with number of units $q_1, q_2,\dots, q_{\ell}>0$ respectively, and sold only at prices $b_1\ge b_2\ge \dots\ge b_{k}\ge 0$ with number of units $r_1, r_2,\dots, r_{k}>0$ respectively, where $b_1\le a_1$ reflects the bid-ask spread and $\ell$ and $k$ belong to $\N\cup\{+\infty\}$. The possibility of $\ell, k=\infty$ allows for infinitely many buy/sell prices in the order book.

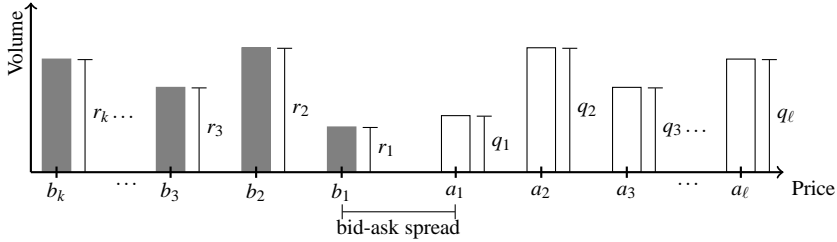
\begin{figure}[h!]
  \centering

\begin{tikzpicture}[scale=0.75]

\filldraw [gray] (-5,0) rectangle (-4.5,2);
\draw [-|]  (-4.25,0) -- (-4.25,1) node[right](-4.25,1) {$r_{k}\dots$} -- (-4.25,2);
\draw [thick] (-4.75,-.1) -- (-4.75,.1);
\node at (-4.75,-.35) {$b_{k}$};
\node   at (-3.5,-.2) {$\dots$};

\filldraw [gray] (-3,0) rectangle (-2.5,1.5);
\draw [-|]  (-2.25,0) -- (-2.25,.75) node[right](-2.25,.75) {$r_{3}$} -- (-2.25,1.5);
\draw [thick] (-2.75,-.1) -- (-2.75,.1);
\node at (-2.75,-.35) {$b_{3}$};

\filldraw [gray] (-1.5,0) rectangle (-1,2.2);
\draw [-|]  (-.75,0) -- (-.75,1.1) node[right](-.75,1.1) {$r_{2}$} -- (-.75,2.2);
\draw [thick] (-1.25,-.1) -- (-1.25,.1);
\node at (-1.25,-.35) {$b_{2}$};

\filldraw [gray] (0,0) rectangle (0.5,.8);
\draw [-|]  (.75,0) -- (.75,.4) node[right](.75,.4) {$r_{1}$} -- (.75,.8);
\draw [thick] (.25,-.1) -- (.25,.1);
\node at (.25,-.35) {$b_{1}$};

\draw  (2,0) rectangle (2.5,1);
\draw [-|]  (2.75,0) -- (2.75,.5) node[right](2.75,.5) {$q_{1}$} -- (2.75,1);
\draw [thick] (2.25,-.1) -- (2.25,.1);
\node at (2.25,-.35) {$a_{1}$};

\draw  (3.5,0) rectangle (4,2.2);
\draw [-|]  (4.25,0) -- (4.25,1.1) node[right](4.25,1.1) {$q_{2}$} -- (4.25,2.2);
\draw [thick] (3.75,-.1) -- (3.75,.1);
\node at (3.75,-.35) {$a_{2}$};

\draw  (5,0) rectangle (5.5,1.5);
\draw [-|]  (5.75,0) -- (5.75,.75) node[right](5.75,.75) {$q_{3}\dots$} -- (5.75,1.5);
\draw [thick] (5.25,-.1) -- (5.25,.1);
\node at (5.25,-.35) {$a_{3}$};
\node   at (6.35,-.2) {$\dots$};

\draw  (7,0) rectangle (7.5,2);
\draw [-|]  (7.75,0) -- (7.75,1) node[right](7.75,1) {$q_{\ell}$} -- (7.75,2);
\draw [thick] (7.25,-.1) -- (7.25,.1);
\node at (7.25,-.35) {$a_{\ell}$};

\draw [|-|]  (.25,-.7) -- (2.25,-.7);
\node   at (1.25,-1) {bid-ask spread};
\draw [thick,->]  (-5.2,0) -- (8,0) node[anchor=north west] {Price};
\draw [thick,->]  (-5.2,0) -- (-5.2,3)node[rotate=90,anchor=south east] {Volume};
\end{tikzpicture}
\caption{A limit order book of $\psi_i$}
\label{fig:bid-ask-chart}
\end{figure}
 
We keep track of $Q_m:= \sum_{j=1}^m q_j$, the total number of units that can be bought at or below the price $a_m$, for $m=1,\dots,\ell$. Similarly, $R_m:=\sum_{j=1}^m r_j$ is the total number of units that can be sold at or above the price $b_m$, for all $m=1,\dots,k$. The total cost $c_i(\eta)$ of trading $\eta$ units of $\psi_i$ is then given by
\begin{equation*}
c_i(\eta)=
\begin{cases}
\sum_{m=1}^{u-1} a_m q_m + a_u(\eta-Q_{u-1})\ge0  &\hbox{if}\ \eta\in(Q_{u-1},Q_{u}],\ u=1,\dots,\ell+1.\\
0 &\hbox{if}\ \eta=0,\\
-\sum_{m=1}^{u-1} b_m r_m + b_u(\eta+R_{u-1})\le0  &\hbox{if}\ \eta\in[-R_{u},-R_{u-1}),\ u=0,\dots,k+1.\\
\end{cases} 
\end{equation*}
where we set $Q_0=R_0=0$, $Q_{\ell+1}=R_{k+1}=\infty$, $a_{\ell+1}=\infty$, $b_{k+1}=0$, and use the convention that $0\cdot\infty=0$ and $\sum_{m=1}^0 = 0$. As shown in Figure~\ref{fig:price impact graph}, $\eta\mapsto c_i(\eta)$ satisfies (C). In particular, $c_i$ is linear on $\R$ if and only if $b_1=a_1\ \hbox{and}\ q_1=r_1=\infty$; this means  
that $\psi_i$ can be traded liquidly at the price $a_1=b_1$, which is the slope of $c_i$.

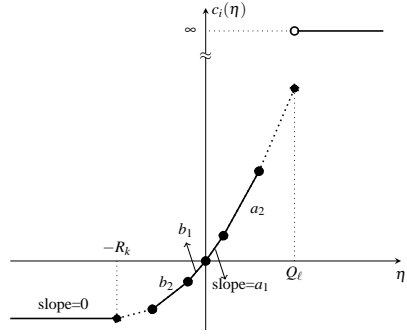
\begin{figure}[h!]
  \centering

\begin{tikzpicture}[scale=.75]

\begin{axis}[axis lines=middle,
	xmin=-11,xmax=11,
         ytick=\empty,xtick=\empty,
         ymin=-3,ymax=11,  
         x label style={at={(axis description cs:1,.2)},anchor=north},
  	 y label style={at={(axis description cs:.56,.95)},anchor=south},
	 ylabel=$c_i(\eta)$,xlabel=$\eta$]
\addplot[samples at={-5,-3,-1,0,1,3,5},thick,dotted,mark=*] {.1*(x+5)^2-2.5};
\addplot[samples at={-3,-1,0,1,3},mark=*,thick] {.1*(x+5)^2-2.5};
\addplot[samples at={5},mark=o,thick] {10};
\addplot[samples at={5.15,10},mark=none,thick] {10};
\addplot[samples at={-5,-11},mark=none,thick] {-2.5};
\addplot[samples at={4.85,0},mark=none,dotted]{10};

\node at (axis cs:-.2,10) [left] {$\infty$};

\node at (axis cs:2,-.5) [below]{slope=$a_1$};
\draw[->] (axis cs:.55,.5) -- (axis cs:1.1,-.8);

\node at (axis cs:2.2,2.2) [right]{$a_2$};

\node at (axis cs:-1.1,.8) [above]{$b_1$};
\draw[->] (axis cs:-.55,-.5) -- (axis cs:-1.1,.8);

\node at (axis cs:-2.2,-1.5) [above]{$b_2$};

\node at (axis cs:-8,-2.5) [above]{slope=0};

\node at (axis cs:5,0) [below]{$Q_{\ell}$};
\draw[dotted] (axis cs:5,0) -- (axis cs:5,7.5);

\node at (axis cs:-5,0) [above]{$-R_{k}$};
\draw[dotted] (axis cs:-5,0) -- (axis cs:-5,-2.5);

\filldraw [color=white] (axis cs:-.2,8.9) rectangle (axis cs:.2,9);
\node at (axis cs:0,8.5) [above]{$\approx$};

\end{axis}

\end{tikzpicture}
\caption{Graph of $\eta\mapsto c_i(\eta)$: $a_1,a_2,\dots$ and $b_1, b_2,\dots$ placed on each segment indicates the slope of each segment and matches the prices in the limit order book for volumes $q_1,q_2,\dots$ and $r_1, r_2,\dots$, respectively.}
\label{fig:price impact graph}
\end{figure}
\end{rem}

\begin{rem}\label{rem:bid-ask spread}
Condition (C) captures two important features of the prices of $\psi_i$: (1) the bid-ask spread, formulated as $[c'_i(0-),c'_i(0+)]$; (2) the non-linearity, i.e. the unit price $\eta\mapsto p_i(\eta)$ is non-constant. This setting in particular allows for zero spread when $c'_i(0-)=c'_i(0+)$, while at the same time the limit order book may induce non-linear pricing. This happens to a highly liquid  asset for which the bid-ask spread is negligible, but transaction cost becomes significant for large trading volumes. 
Also, $c_i$ is linear if and only if $\psi_i$ can be traded liquidly, with whatever units, at one single price $p_i$ (which is the slope of $c_i$). 

Note that \cite{BZZ14} has recently considered bid-ask spreads, but not non-linear pricing, of hedging options under model uncertainty. In a model-independent setting, while a non-linear pricing operator for hedging options has been used in \cite{DS14}, the non-linearity does not reflect the non-constant unit price of an option in its limit order book (see \cite[(2.3)]{DS14}); instead, it captures a market where the price of a portfolio of options may be lower than the sum of the respective prices of the options (see the second line in the proof of \cite[Lemma 2.4]{DS14}).    
\end{rem}

\subsection{Constrained trading strategies}

\begin{defn}[Trading strategies]
We say $\Delta=\{\Delta_t\}_{t=0}^{T-1}$ is a trading strategy if $\Delta_0\in\R^d$ is a constant and $\Delta_t:(\R^d_+)^t\mapsto \R^d$ is Borel measurable for all ${t=1,\dots, T-1}$. Moreover, the stochastic integral  of $\Delta$ with respect to ${x=(x_1,\dots,x_T)\in (\R^d_+)^T}$ will be expressed as
\[
(\Delta\bcdot x)_t := \sum_{i=0}^{t-1}\Delta_i(x_1,\dots,x_i) \cdot (x_{i+1}-x_i),\ \ \hbox{for}\ t=1,\dots, T,
\]
where in the right hand side above, $\Delta_i=(\Delta^1_i,\dots,\Delta^d_i)$, $x_i=(x_i^1,\dots,x^d_i)$, and `` $\cdot$ '' denotes the inner product in $\R^d$. We will denote by $\mathcal{H}$ the collection of all trading strategies. 
Also, for any collection $\mathcal{J}\subseteq\mathcal{H}$, we introduce the sub-collections
\begin{equation}\label{defn cS^infty}
\begin{split}
\mathcal{J}^\infty &:=\{\Delta\in\mathcal{J}:\ \Delta_t:(\R^d_+)^t\mapsto\R^d\ \hbox{is bounded},\ \forall\ t=1,\dots,T-1\},\\
\mathcal{J}^\infty_c &:=\{\Delta\in\mathcal{J}^\infty:\ \Delta_t:(\R^d_+)^t\mapsto\R^d\ \hbox{is continuous},\ \forall\ t=1,\dots,T-1\}.
\end{split}
\end{equation}
\end{defn}

In this paper, we require the trading strategies to lie in a sub-collection $\mathcal{S}$ of $\mathcal{H}$, prescribed as below.  

\begin{defn}[Adaptively convex portfolio constraint]\label{defn:cS}
$\mathcal{S}$ is a set of trading strategies such that
\begin{itemize}
\item [(i)] $0\in\mathcal{S}$.
\item [(ii)] For any $\Delta,\Delta'\in\mathcal{S}$ and adapted process $h$ with $h_t\in[0,1]$ for all $t=0,\dots, T-1$, 
\[
\{h_t\Delta_t + (1-h_t)\Delta'_t\}_{t=0}^{T-1}\in \mathcal{S}.
\]
\item [(iii)] For any $\Delta\in\cS^\infty$, $\Q\in\Pi$, and $\eps>0$, there exist a closed set $D_\eps\subseteq(\R^d_+)^T$ and $\Delta^\eps\in\cS^\infty_c$ such that
\[
\Q(D_\eps)>1-\eps\ \ \hbox{and}\ \ \Delta_t = \Delta^\eps_t\ \hbox{on}\ D_\eps\ \hbox{for}\ t=0,\dots,T-1.
\] 
\end{itemize}
\end{defn}

\begin{rem}
In Definition~\ref{defn:cS}, (i) and (ii) are motivated by \cite[Section 9.1]{FS-book-11}, while (iii) is a technical assumption which allows us to perform continuous approximation in Lemma~\ref{lem:S->S_c}. This approximation in particular enables us to establish the superhedging duality in Proposition~\ref{thm:duality}. 
In fact, if we only have conditions (i) and (ii), the duality in Proposition~\ref{thm:duality} may fail in general, as demonstrated in Appendix~\ref{sec:appendix}.
\end{rem}

An explained below, Definition~\ref{defn:cS} (iii) covers a large class of convex constraints.

\begin{rem} [Adapted convex constraints]\label{rem:adapted convex}
Let $\{K_t\}_{t=0}^{T-1}$ be an adapted set-valued process such that for each $t$, $K_t$ maps $(x_1,\dots,x_t)\in(\R^d_+)^t$ to a closed convex set $K_t(x_1,\dots,x_t)\subseteq\R^d$ which contains $0$.
Consider the collection of trading strategies
\[
\cS:=\{\Delta\in\mathcal{H} : \hbox{for each}\ t\ge 0,\ \Delta_t(x_1,\dots,x_t)\in K_t(x_1,\dots,x_t),\ \forall\ (x_1,\dots,x_t)\in(\R^d_+)^t\},
\]  
which satisfies Definition~\ref{defn:cS} (i) and (ii) trivially. To obtain Definition~\ref{defn:cS} (iii), we assume additionally that for each $t\ge 1$, the set-valued map $K_t:(\R^d_+)^t\mapsto 2^{\R^d}$ is lower semicontinuous, in the sense that
\begin{equation}\label{lower semiconti.}
\hbox{for any}\ V\ \hbox{open in}\ \R^d,\ \hbox{the set}\ \{x\in(\R^d_+)^t: K_t(x)\cap V\neq\emptyset\}\ \hbox{is open in}\ (\R^d_+)^t.   
\end{equation}
This is equivalent to the following condition:
\begin{equation}\label{lower semiconti.'}
\begin{split}
&\forall\ y_0\in K_t(x_0)\ \hbox{and}\ \{x_n\}\subset (\R^d_+)^t\ \hbox{such that}\  x_n\to x_0,\\
 &\exists\ y_n\in K_t(x_n)\ \hbox{such that}\ y_n\to y_0;   
\end{split}
\end{equation}
see e.g. \cite[Definition 1.4.2]{AF09} and the remark below it, and \cite[Section 2.5]{GRTZ-book-03}. 

To check (iii), let us fix $\Delta\in\cS^\infty$, $\Q\in\Pi$, and $\eps>0$. For each $t=1,\dots,T-1$, by Lusin's theorem, there exists a closed set $D_{\eps,t}\subset(\R^d_+)^t$ such that $\Delta_t|_{D_{\eps,t}}:D_{\eps,t}\mapsto K_t$ is continuous and ${\Q(D_{\eps,t}\times(\R^d_+)^{T-t})> 1-\frac{\eps}{T-1}}$. Under \eqref{lower semiconti.}, we can apply the theory of continuous selection (see e.g. \cite[Theorem 3.2$''$]{Michael56}) to find a bounded continuous function $\Delta^\eps_t:(\R^d_+)^t\mapsto K_t$ such that $\Delta^\eps_t=\Delta_t$ on $D_{\eps,t}$. Now, set $\Delta^\eps_0=\Delta_0$, and define $D_\eps:=\bigcap_t D_{\eps,t}\times(\R^d_+)^{T-t}$, which is by definition closed in $\R^T_+$. We see that $\Delta^\eps\in\cS^\infty_c$, $\Q(D_\eps)>1-\eps$, and $\Delta_t=\Delta^\eps_t$ on $D_{\eps}$ for all $t=0,\dots,T-1$. This already verifies Definition~\ref{defn:cS} (iii). 

Note that for the special case where $\{K_t\}_{t=1}^T$ is deterministic, $K_t$ is a fixed subset of $\R^d$ for each $t$ and thus \eqref{lower semiconti.'} is trivially satisfied. See Examples~\ref{eg:shortselling} and \ref{eg:drawdown} below for a concrete illustration of deterministic and adapted convex constraints.
\end{rem}


\section{The superhedging duality}\label{sec:duality}
For a path-dependent exotic option with payoff function $\Phi:(\R^d_+)^T\mapsto\R$, we intend to construct a semi-static superhedging portfolio, which consists of three parts: static positions in vanilla calls, static positions in $\{\psi_i\}_{i\in I}$, and a dynamic trading strategy $\Delta\in\cS$. More precisely, consider
\begin{align*}
\mathcal{C}&:=\bigl\{\varphi:\R\mapsto\R:\varphi(x)  =a+\sum_{i=1}^n b_i (x-K_i)^+\ \hbox{with}\ a\in\R,\ n\in\N,\ b_i\in\R,\ K_i> 0\bigr\},
\\
\mathcal{R}^I&:=\{\eta=(\eta_i)_{i\in I}\in\R: \eta_i\neq 0\ \hbox{for finitely many}\ i\hbox{'s}\}.
\end{align*}
We intend to find $u=\{u^n_t\in\mathcal{C}:n=1,\dots,d;\ t=1,\dots,T\}$, $\eta\in\mathcal{R}^{I}$,  and $\Delta\in\cS$ such that 
\begin{equation}\label{pathwise superhedge}
\Psi_{u,\eta,\Delta}(x):=\sum_{t=1}^T\sum_{n=1}^d u^n_t(x^n_t) + \sum_{i\in I}(\eta_i\psi_i - c_i(\eta_i))+ (\Delta\bcdot x)_T \ge \Phi(x)\ \ \forall x\in(\R^d_+)^T,
\end{equation}
where we assume that $0\cdot\infty=0$. In the definition of $\mathcal{C}$, we specifically require $K_i$ to be strictly positive. This is because $K_i=0$ corresponds to trading the risky assets, which is already incorporated into $\Delta\in\cS$ and should not be treated as part of the static positions. 
By setting $\cU$ as the collection of all ${u=\{u^n_t\in\mathcal{C}: n=1,\dots,d,\ t=1,\dots,T\}}$, we define the superhedging price of $\Phi$ by
\begin{equation}\label{def:D}
\begin{split}
D(\Phi)&:=\inf\Bigl\{\sum_{t=1}^T\sum_{n=1}^d\int_{\R_+}u^n_t d\mu^n_t: u\in\cU,\ \eta\in\mathcal{R}^I\ \hbox{and}\ \Delta\in\cS\ \\
&\hspace*{5cm}\hbox{satisfy}\ \Psi_{u,\eta,\Delta}\ge \Phi,\ \forall x\in(\R^d_+)^T\Bigr\}.
\end{split}
\end{equation}
By introducing $\cU_0:=\{u\in\cU: \sum_{t=1}^T\sum_{n=1}^d\int_{\R_+}u^n_td\mu^n_t=0\}$, we may express \eqref{def:D} as
\[
D(\Phi)=\inf\bigl\{a\in\R:  u\in\cU_0, \eta\in\mathcal{R}^I  \hbox{and } \Delta\in\cS\  \hbox{satisfy}\ a+\Psi_{u,\eta,\Delta}\ge \Phi,
 \forall x\in(\R^d_+)^T\bigr\} 
\]
Our goal in this section is to derive a superhedging duality associated with $D(\Phi)$.

\subsection{Upper variation process}

In order to deal with the portfolio constraint $\cS$, we introduce an auxiliary process for each $\Q\in\Pi$, as suggested in \cite[Section 9.2]{FS-book-11}.

\begin{defn}\label{defn:A^Q_t}
Given $\Q\in\Pi$, the upper variation process $A^\Q$ for $\cS$ is defined by 
\[
A^\Q_0:=0,\ \ \hbox{and}\ \ A^\Q_{t+1}-A^\Q_t:= \sideset{}{^\Q}\esssup_{\Delta\in\cS}\left\{\Delta_t \cdot (\E^\Q[S_{t+1}\mid\cF_t]-S_t) \right\}, \ \ t=0,\dots,T-1.
\]
\end{defn}
First, Note that the conditional expectation in the definition of $A^\Q$ is well-defined, thanks to Remark~\ref{rem:E[S] finite}. Next, since Definition \ref{defn:cS} (i)-(ii) implies  $1_{\{|\Delta|\le n\}} \Delta\in \cS$ whenever $\Delta\in\cS$, we may replace $\cS$ by $\cS^\infty$ in the above definition. It follows that
\begin{equation}\label{increment A^Q_t}
A^\Q_{t+1}-A^\Q_t = \sideset{}{^\Q}\esssup_{\Delta\in\cS^\infty}\E^\Q[\Delta_t \cdot (S_{t+1}-S_t)\mid\cF_t], \ \ t=0,\dots,T-1.
\end{equation}
Therefore,
\begin{equation}\label{A^Q_t}
A^\Q_t = \sum_{i=1}^{t} \sideset{}{^\Q} \esssup_{\Delta\in\cS^\infty}\E^\Q[\Delta_{i-1} \cdot (S_{i}-S_{i-1})\mid\cF_{i-1}],\ \  t=1,\dots,T. 
\end{equation}

\begin{lem}\label{lem:E[A_t]} 
For any $\Q\in\Pi$ and $t=1,\dots,T$, we have
\[
\E^\Q\left[\sideset{}{^\Q} \esssup_{\Delta\in\cS^\infty}\E^\Q[\Delta_{t} \cdot (S_{t+1}-S_{t})\mid\cF_{t}]\right] = \sup_{\Delta\in\cS^\infty}\E^\Q[\Delta_{t} \cdot (S_{t+1}-S_t)].
\]
This in particular implies that 
\begin{equation}\label{E[A_t]}
\E^\Q[A^\Q_t] = \sup_{\Delta\in\cS^\infty} \E^\Q[(\Delta\bcdot S)_t].  
\end{equation}
\end{lem}
\begin{proof}
First, note that  $\{\E^\Q[\Delta_{t} \cdot (S_{t+1}-S_{t})\mid\cF_{t}]:\Delta\in\cS^\infty\}$ is directed upwards. Indeed, given $\Delta,\Delta'\in\cS^\infty$, define $\tilde{\Delta}_s:=\Delta_s 1_{\{s\neq t\}}+(\Delta_s 1_A+\Delta'_s 1_{A^c})1_{\{s=t\}}$, where
\[
A:=\{\E^\Q[\Delta_t \cdot (S_{t+1}-S_t)\mid\cF_t]\ge \E^\Q[\Delta'_{t} \cdot (S_{t+1}-S_t)\mid\cF_t]\}\in\cF_t.
\] 
Then, $\tilde{\Delta}\in\cS^\infty$ by Definition~\ref{defn:cS} (ii), and 
\[
E^\Q[\tilde{\Delta}_t(S_{t+1}-S_t)\mid\cF_t]=\max\{\E^\Q[\Delta_t \cdot (S_{t+1}-S_t)\mid\cF_t], \E^\Q[\Delta'_t \cdot (S_{t+1}-S_t)\mid\cF_t]\}.
\]
We can therefore apply \cite[Theorem A.33, pg. 496]{FS-book-11} and get 
\begin{align*}
\E^\Q \left[\sideset{}{^\Q} \esssup_{\Delta\in\cS^\infty}\E^\Q [\Delta_{t} \cdot (S_{t+1}-S_{t})\mid\cF_{t}] \right] = \sup_{\Delta\in \cS^\infty}\E^\Q[\Delta_{t} \cdot (S_{t+1}-S_t)].
\end{align*}
Now, in view of \eqref{A^Q_t}, we have
\[
\E^\Q[A^\Q_t]=\sum_{i=1}^{t}\sup_{\Delta\in \cS^\infty}\E^\Q[\Delta_{i-1} \cdot (S_{i}-S_{i-1})]=\sup_{\Delta\in \cS^\infty}\E^\Q[(\Delta\bcdot S)_t],
\]
where the last equality follows from Definition~\ref{defn:cS} (ii). 
\qed\end{proof}

On strength of Definition~\ref{defn:cS} (iii), we can actually replace $\cS^\infty$ by $\cS^\infty_c$ in \eqref{E[A_t]}.

\begin{lem}\label{lem:S->S_c}
For each $t=1,\dots,T$, 
\begin{equation*}\label{S->S_c}
\E^\Q[A^\Q_t] = \sup_{\Delta\in\cS^\infty_c} \E^\Q[(\Delta\bcdot S)_t].  
\end{equation*}
\end{lem}
\begin{proof}
In view of \eqref{E[A_t]}, it suffices to show that, for each fixed $\Delta\in\cS^\infty$, there exists $\{\Delta^\eps\}_\eps\subset \cS^\infty_c$ such that ${\E^\Q[(\Delta\bcdot S)_T]=\lim_{\eps\to 0}\E^\Q[(\Delta^\eps\bcdot S)_T]}$. Take $M>0$ such that $|\Delta_t|\le M$ for all $t=1,\dots,T-1$. By Definition~\ref{defn:cS} (iii), for any $\eps>0$, there exist $D_\eps$ closed in $\R^T_+$ and $\Delta^\eps\in\cS^\infty_c$ such that $\Q(D_\eps)> 1-\eps$, $\Delta^\eps=\Delta$ on $D_\eps$, and $|\Delta^\eps_t|\le M$ for all $t=1,\dots,T-1$. It follows that
\begin{align*}
\left|\E^\Q[(\Delta\bcdot S)_T]-\E^\Q[(\Delta^\eps\bcdot S)_T]\right|&\le \E^\Q\left[\left|\left((\Delta-\Delta^\eps)\bcdot S\right)_T\right| 1_{D^c_\eps}\right]\\
&\le \E^\Q\Bigl[2M\sum_{t=0}^{T-1}|S_{t+1}-S_t| 1_{D^c_\eps}\Bigr].
\end{align*}
Thanks to Remark~\ref{rem:E[S] finite}, the random variable $2M\sum_{t=0}^{T-1}|S_{t+1}-S_t|$ is $\Q$-integrable. We can then conclude from the above inequality that $\E^\Q[(\Delta^\eps\bcdot S)_T]\to\E^\Q[(\Delta\bcdot S)_T]$.
\qed\end{proof}

\begin{defn}\label{def:Q_S}
Let $\cQ_\cS$ be the collection of $\Q\in\Pi$ such that $\E^\Q[A^\Q_T]<\infty$.
\end{defn}

\begin{rem}\label{rem:uniform bdd}
If strategies in $\cS$ are uniformly bounded, i.e. $\exists$ $c>0$ such that $|\Delta|\le c$ for all $\Delta\in\cS$, then we deduce from \eqref{E[A_t]} and Remark~\ref{rem:E[S] finite} that $\cQ_\cS=\Pi$. 
\end{rem}

\begin{lem}\label{lem:local supermart.}
Given any $\Delta\in\cS$, the process $(\Delta\bcdot S)_t-A^\Q_t$ is a local $\Q$-supermartingale, for all $\Q\in\cQ_\cS$.
\end{lem}

\begin{proof}
This result follows from the argument in \cite[Proposition 9.18]{FS-book-11}. We present the proof here for completeness. Consider the stopping time
\[
\tau_n := \inf\{t\ge 0 : |\Delta_t|>n\ \hbox{or}\ \E^\Q[|S_{t+1}-S_t|\mid\cF_t]>n\}\wedge T,
\]
where the conditional expectation is well-defined thanks to Remark~\ref{rem:E[S] finite}. Given a process $V_t$, let us denote by $V^{n}_t$ the stopped process $V_{t\wedge\tau^n}$. Observe that 
\[
|(\Delta\bcdot S)^{n}_{t+1}-(\Delta\bcdot S)^{n}_t|\le 1_{\{\tau_n\ge t+1\}}|\Delta_{t}| \ |S_{t+1}-S_t|.
\]
Thanks again to Remark~\ref{rem:E[S] finite}, this implies that $(\Delta\bcdot S)^{n}_t$ is $\Q$-integrable. Moreover,
\[
\E^\Q[(\Delta\bcdot S)^{n}_{t+1}-(\Delta\bcdot S)^{n}_t\mid\cF_t]=1_{\{\tau_n\ge t+1\}}\Delta_{t}\cdot(\E^\Q[S_{t+1}\mid\cF_t]-S_t)\le (A^\Q)^{n}_{t+1}-(A^\Q)^{n}_t,
\]
where the inequality follows from \eqref{increment A^Q_t}. Since $\E^\Q[A^\Q_T]<\infty$, the above inequality shows that $(\Delta\bcdot S)^{n}_t-(A^\Q)^{n}_t$ is a $\Q$-supermartingale.
\qed\end{proof}

With some integrability at the terminal time $T$, the local supermartingale in the above result becomes a true supermartingale.

\begin{lem}\label{lem:supermart.}
Fix $\Delta\in\cS$ and $\Q\in\cQ_\cS$. If $\exists$ $\Q$-integrable random variable $\varphi$ such that $(\Delta\bcdot S)_T\ge \varphi$ $\Q$-a.s., then $(\Delta\bcdot S)_t-A^\Q_t \ge \E^\Q[\varphi-A^\Q_T\mid\cF_{t}]$ $\Q\hbox{-a.s.}$ for all $t=1,\dots,T$. 
This in particular implies that $(\Delta\bcdot S)_t-A^\Q_t$ is a true $\Q$-supermartingale.
\end{lem}

\begin{proof}
Using the same notation as in the proof of Lemma~\ref{lem:local supermart.}, we know that there exist a sequence $\{\tau_n\}$ of stopping times such that $\tau^n\uparrow\infty$ $\Q$-a.s. and the stopped process $(\Delta\bcdot S)^n_t-(A^\Q)^n_t$ is a $\Q$-supermartingale, for each $n\in\N$.  
We will prove this lemma by induction. Given any $t=2,\dots,T$ such that ${(\Delta\bcdot S)_t-A^\Q_t\ge \E^\Q[\varphi-A^\Q_T\mid\cF_{t}]}$, we obtain from the supermartingale property that 
\begin{align*}
0&\ge 1_{\{t-1<\tau_n\}}\E^\Q\left[\left((\Delta\bcdot S)^n_t-(A^\Q)^n_t\right)-\left((\Delta\bcdot S)^n_{t-1}-(A^\Q)^n_{t-1}\right) \mid \cF_{t-1}\right]\\
&= \E^\Q\left[1_{\{t-1<\tau_n\}}\left\{ \left((\Delta\bcdot S)_t-A^\Q_t\right)-\left((\Delta\bcdot S)_{t-1}-A^\Q_{t-1}\right) \right\} \mid \cF_{t-1}\right]\\
&\ge \E^\Q\left[1_{\{t-1<\tau_n\}}\left\{ \E^\Q[\varphi-A^\Q_T\mid\cF_{t}]-\left((\Delta\bcdot S)_{t-1}-A^\Q_{t-1}\right) \right\} \mid \cF_{t-1}\right]\\
&= 1_{\{t-1<\tau_n\}} \E^\Q[\varphi-A^\Q_T\mid\cF_{t-1}]- 1_{\{t-1<\tau_n\}}\left((\Delta\bcdot S)_{t-1}-A^\Q_{t-1}\right).
\end{align*}
Sending $n\to\infty$, we conclude that $(\Delta\bcdot S)_{t-1}-A^\Q_{t-1}\ge \E^\Q[\varphi-A^\Q_T\mid\cF_{t-1}]$ $\Q$-a.s. Now, by Lemma~\ref{lem:local supermart.}, $(\Delta\bcdot S)_t-A^\Q_t$ is a local $\Q$-supermartingale bounded from below by a martingale, and thus a true $\Q$-supermartingale (see e.g. \cite[Proposition 9.6]{FS-book-11}).
\qed\end{proof}


\subsection{Derivation of the superhedging duality} In view of the static holdings of $\{\psi_i\}_{i\in I}$ in \eqref{pathwise superhedge}, we introduce 
\begin{equation}\label{cE^Q_I}
\cE^\Q_I := \sup_{\eta\in \mathcal{R}^I}\sum_{i\in I}(\eta_i\E^\Q[\psi_i]-c_i(\eta_i))\ge 0\ \ \ \hbox{for}\ \Q\in\Pi.
\end{equation}
Set $F(I):=\{J\subseteq I: J\ \hbox{is a finite set}\}$. We observe that
\begin{equation}\label{cE^Q_I=...}
\cE^\Q_I = \sup_{J\in F(I)}\sup_{\eta\in\R^{|J|}}\sum_{i\in J}(\eta_i\E^\Q[\psi_i]-c_i(\eta_i))=\sup_{J\in F(I)}\sum_{i\in J}\sup_{\eta\in\R}(\eta\E^\Q[\psi_i]-c_i(\eta)).
\end{equation}
Consider the collection of measures
\begin{equation}\label{Q_S,I}
\cQ_{\cS,I}:=\{\Q\in\cQ_\cS : \cE^\Q_I<\infty\}.
\end{equation}

\begin{rem}\label{rem:cE^Q_I finite}
Fix $\Q\in\Pi$. For any $i\in I$, suppose the following two conditions hold. 
\begin{itemize}
\item [(i)] $c_i(\eta)=\infty$ for some $\eta>0$ or $ \E^\Q[\psi_i] < c'_i(\infty)$,
\item [(ii)] $c_i(\eta)=\infty$ for some $\eta<0$ or $ \E^\Q[\psi_i] > c'_i(-\infty)$.
\end{itemize}
By the convexity of $\eta\mapsto c_i(\eta)$, we have $\sup_{\eta\in\R}(\eta\E^\Q[\psi_i]-c_i(\eta))<\infty$.
Thus, in view of \eqref{cE^Q_I=...}, if $I$ is a finite set, and (i)-(ii) above are satisfied for all $i\in I$, then $\cE^\Q_I<\infty$.
\end{rem}

We will work on deriving a duality between $D(\Phi)$ defined in \eqref{def:D} and
\begin{equation}\label{def:P}
P(\Phi) := \sup_{\Q\in\cQ_{\cS,I}}\{\E^\Q[\Phi-A^\Q_T]-\cE^\Q_I\}.
\end{equation}
The following minimax result, taken from \cite[Corollary 2]{Terkelsen72}, will be useful.

\begin{lem}\label{lem:minimax}
Let $X$ be a compact convex subset of a topological vector space, $Y$ be a convex subset of a vector space, and $f:X\times Y\mapsto\R$ be a function satisfying
\begin{itemize}
\item [(i)] For each $x\in X$, the map $y\mapsto f(x,y)$ is convex on $Y$.
\item [(ii)] For each $y\in Y$, the map $x\mapsto f(x,y)$ is upper semicontinuous and concave on $X$.
\end{itemize} 
Then, 
$\inf_{y\in Y} \sup_{x\in X} f(x,y) = \sup_{x\in X} \inf_{y\in Y} f(x,y)$.
\end{lem}

Let us first derive a superhedging duality for the case where $I=\emptyset$, i.e. no option is tradable at time $0$ except vanilla calls. The pathwise relation in \eqref{pathwise superhedge}  reduces to
\[
\Psi_{u,\Delta}(x):= \sum_{t=1}^T\sum_{n=1}^d u^n_t(x^n_t) + (\Delta\bcdot x)_T\ge \Phi(x),\ \forall x\in(\R^d_+)^T.
\]
By the convention that the sum over an empty set is $0$, $D(\Phi)$ in \eqref{def:D} becomes
\begin{equation*}
\begin{split}
D_\emptyset(\Phi) & := \inf\Bigl\{\sum_{t=1}^T\sum_{n=1}^d\int_{\R_+}u^n_t d\mu^n_t: u\in\cU\ \hbox{and}\  \Delta\in\cS\ \hbox{such that}\ \Psi_{u,\Delta}(x)\ge \Phi(x),\\
&\hspace*{8.5cm}  \forall x\in(\R^d_+)^T\Bigr\}.\label{D I=empty}
\end{split}
\end{equation*}
Also, since $I=\emptyset$ implies that $F(I)=\{\emptyset\}$, we deduce from \eqref{cE^Q_I=...} that $\cE^\Q_I=0$, as it is a summation over an empty set. It follows that $P(\Phi)$ in \eqref{def:P} reduces to
\begin{equation}
P_\emptyset(\Phi):=\sup_{\Q\in\cQ_\cS}\E^\Q[\Phi-A^\Q_T]\label{P I=empty}.
\end{equation}

\begin{prop}\label{thm:duality}
Let $I=\emptyset$. Suppose $\Phi:(\R^d_+)^T\mapsto\R$ is measurable and $\exists$ $K>0$ such that
\begin{equation}\label{Phi linear grow}
\Phi(x_1,\dots,x_T)\le K \left(1+\sum_{t=1}^T\sum_{n=1}^d x^n_t\right),\ \ \hbox{for all}\ x\in(\R^d_+)^T.
\end{equation}
\begin{itemize}
\item [(i)] We have $P_\emptyset(\Phi)\le D_\emptyset(\Phi)$.
\item [(ii)] If $\Phi$ is upper semicontinuous, then $P_\emptyset(\Phi)=D_\emptyset(\Phi)$.
\item [(iii)] If $\Phi$ is upper semicontinuous and $\cQ_\cS\neq\emptyset$, there exists $\Q^*\in\cQ_\cS$ such that 
$P_\emptyset(\Phi)=\E^{\Q^*}[\Phi-A^{\Q^*}_T]$.
\end{itemize}
\end{prop}

\begin{proof}
 First, by Remark~\ref{rem:E[S] finite}, \eqref{Phi linear grow}, and Definition~\ref{def:Q_S}, $P_\emptyset$ is indeed well defined.\\
(i) Take $u\in\cU$ and $\Delta\in\cS$ such that $\Psi_{u,\Delta}\ge \Phi$. For any $\Q\in\cQ_\cS$, note that 
\[
(\Delta\bcdot S)_T\ge \Phi(x)-\sum_{t=1}^T\sum_{n=1}^d u^n_t(x^n_t).
\]
If $\E^\Q[\Phi^-]<\infty$, then $\Phi(x)-\sum_{t=1}^T\sum_{i=1}^du^i_t(x_t)$ is $\Q$-integrable thanks to \eqref{Phi linear grow}. We then conclude from Lemma~\ref{lem:supermart.} that $(\Delta\bcdot S)_t-A^\Q_t$ is a true $\Q$-supermartingale. Hence,
\begin{equation}\label{E<D}
\E^\Q[\Phi-A^\Q_T]\le\E^\Q\left[\sum_{t=1}^T\sum_{n=1}^d u^n_t(S^n_t) + (\Delta\bcdot S)_T -A^\Q_T\right]\le \sum_{t=1}^T \sum_{n=1}^d \int_{\R_+}u^n_t d\mu^n_t.
\end{equation}
If $\E^\Q[\Phi^-]=\infty$, then \eqref{E<D} trivially holds. By taking supremum over $\Q\in\cQ_\cS$ and using the arbitrariness of $u$, we obtain from \eqref{E<D} the desired inequality. \\
(ii) We will use an argument similar to \cite[equations (3.1)-(3.4)]{BHP13}. First, observe that
\begin{align}
D_\emptyset(\Phi) \le & \inf\left\{\sum_{t=1}^T\sum_{n=1}^d\int_{\R_+}u^n_td\mu^i_t:u\in\cU \hbox{and } \Delta\in\cS^\infty_c \hbox{such that}\ \Psi_{u,\Delta}(x)\ge \Phi(x)\right\}\nonumber\\
=\inf_{\Delta\in\cS^\infty_c}&\inf\biggl\{\sum_{t=1}^T\sum_{n=1}^d\int_{\R_+}u^n_td\mu^n_t:u\in\cU \hbox{such that }\! \sum_{t=1}^T\sum_{n=1}^du^n_t(x^n_t)\ge\Phi(x)-(\Delta\bcdot x)_T\biggr\}\nonumber\\
=\inf_{\Delta\in\cS^\infty_c}&\sup_{\Q\in\Pi}\E^\Q[\Phi(x)-(\Delta\bcdot x)_T].   \label{1}   
\end{align}
Here, \eqref{1} follows from the theory of optimal transport (see e.g. \cite[Proposition 2.1]{BHP13}), which requires the upper semicontinuity of $\Phi$. Now, we intend to apply Lemma \ref{lem:minimax} to \eqref{1}, with $X=\Pi$, $Y=\cS^\infty_c$, and $f(\Q,\Delta)=\E^\Q[\Phi(x)-(\Delta\bcdot x)_T]$. The only condition in Lemma~\ref{lem:minimax} which is not obvious is the upper semicontinuity of $\Q\mapsto f(\Q,\Delta)$. For each $\Delta\in\cS^\infty_c$, thanks to \eqref{Phi linear grow}, the upper semicontinuous function $\Phi(x)-(\Delta\bcdot x)_T$ is bounded from above by the continuous function 
\begin{equation}\label{upper bound 1}
\ell(x):= K\left(1+\sum_{t=1}^T\sum_{n=1}^d x^n_t\right)+ |\Delta|_\infty\sum_{n=1}^d(x^n_0+2(x^n_1+\dots+x^n_{T-1})+x^n_T),
\end{equation}
where $|\Delta|_\infty:=|\Delta_0|\vee\max\{\sup_{z\in(\R^d_+)^t}|\Delta_t(z)|:t=1,\dots,T-1\}<\infty$. Take any sequence $\{\Q_n\}_{n\in\N}$ in $\Pi$ which converge weakly to some $\Q^*\in\Pi$. Observing that $\Q\mapsto\E^\Q[\ell]$ is a constant function on $\Pi$,
we conclude from \cite[Lemma 4.3]{Villani-book-09} that 
\[
\limsup_{n\to\infty}\E^{\Q_n}[\Phi(x)-(\Delta\bcdot x)_T]\le\E^{\Q^*}[\Phi(x)-(\Delta\bcdot x)_T],
\]
which shows the upper semicontinuity of $\Q\mapsto f(\Q,\Delta)$. Now, applying Lemma \ref{lem:minimax} to \eqref{1} yields
\begin{align}
D_\emptyset(\Phi) &\le\sup_{\Q\in\Pi}\inf_{\Delta\in\cS^\infty_c}\E^\Q[\Phi(x)-(\Delta\bcdot x)_T]=\sup_{\Q\in\Pi}\bigg\{\E^\Q[\Phi] - \sup_{\Delta\in\cS^\infty_c}\E^\Q[(\Delta\bcdot S)_T] \bigg\}\nonumber\\
&=\sup_{\Q\in\Pi}\left\{\E^\Q[\Phi] - \E^\Q[A^\Q_T]\right\} =\sup_{\Q\in\cQ_\cS}\left\{\E^\Q[\Phi] - \E^\Q[A^\Q_T]\right\} = P_\emptyset(\Phi)\nonumber,   
\end{align}
where the second line follows from Lemma~\ref{lem:S->S_c}.\\
(iii) In view of Definition \ref{def:Q_S}, we can write 
${P_\emptyset(\Phi)=\sup_{\Q\in\Pi}\E^\Q[\Phi-A^\Q_T]}$ by replacing $\cQ_\cS$ by $\Pi$ in \eqref{P I=empty}. 
Since $\Pi$ is compact under the topology of weak convergence (Remark~\ref{rem:Pi compact}), it suffices to show that $\Q\mapsto f(\Q):=\E^\Q[\Phi-A^\Q_T]$ is upper semicontinuous. Since the argument in part (ii) already implies that $\Q\mapsto \E^\Q[\Phi]$ is upper semicontinuous, it remains to show that $\Q\mapsto g(\Q):=\E^\Q[A^\Q_T]$ is lower semicontinuous. Similar to \eqref{upper bound 1}, for each $\Delta\in\cS^\infty_c$, we have 
$|(\Delta\bcdot x)_T|\le h(x)$ with $h$ defined by $h(x):=|\Delta|_\infty\sum_{n=1}^d(x^n_0+2(x^n_1+\dots+x^n_{T-1})+x^n_T].
$
For any sequence $\{\Q_n\}_{n\in\N}$ in $\Pi$ which converge weakly to some $\Q^*\in\Pi$, applying \cite[Lemma 4.3]{Villani-book-09} to the functions $(\Delta\bcdot x)_T$ and $-(\Delta\bcdot x)_T$ gives
\[
\liminf_{n\to\infty} \E^{\Q_n}[(\Delta\bcdot x)_T]\ge \E^{\Q^*}[(\Delta\bcdot x)_T]\ge \limsup_{n\to\infty}\E^{\Q_n}[(\Delta\bcdot x)_T].
\] 
It follows that $\Q\mapsto g_\Delta(\Q):=\E^\Q[(\Delta\bcdot S)_T]$ is continuous. Thanks to Lemma~\ref{lem:S->S_c}, we have $g(\Q)= \sup_{\Delta\in\cS_c^\infty}g_\Delta(\Q)$ is lower semicontinuous, as a supremum of continuous functions. \qed
\end{proof}

\begin{rem}
The condition $\cQ_\cS\neq\emptyset$ is not needed for Proposition~\ref{thm:duality} (i) and (ii). Indeed, if $\cQ_\cS=\emptyset$, then $P(\Phi)=-\infty$ and thus part (i) trivially holds; also, the arguments in part (ii) hold true as long as $\Pi\neq\emptyset$, which is guaranteed by Remark~\ref{rem:Pi compact}.
\end{rem}

\begin{rem}\label{rem:S=H}
Proposition~\ref{thm:duality} extends \cite[Theorem 1.1]{BHP13} to the case with portfolio constraints. To see this, consider the no-constraint case, i.e. $\cS=\cH$. Observe that $\cS=\cH$ implies $\cQ_\cS=\cM$, with $\cM$ defined as in \eqref{def:cM}.
While $\cM\subseteq \cQ_\cS$ is obvious, the other inclusion follows from Definition~\ref{defn:A^Q_t}. Indeed, given $\Q\in\cQ_\cS\setminus\cM$, there must exist $t\in\{0,\dots,T-1\}$ such that $\E^\Q[S_{t+1}\mid\cF_t]\neq S_t$. Since $\cS=\cH$, we have $A^\Q_{t+1}-A^\Q_t=\infty$, contradicting $\Q\in\cQ_\cS$. The duality in Proposition~\ref{thm:duality} reduces to
\[
D_\emptyset(\Phi)=P_\emptyset(\Phi) = \sup_{\Q\in\cM}\E^\Q[\Phi],
\]
which recovers \cite[Theorem 1.1]{BHP13}.
\end{rem}

\begin{rem}\label{rem:multi-d}
Proposition~\ref{thm:duality} also extends \cite[Theorem 1.1]{BHP13} to the case with multi-dimensional $S$. Since \cite[Theorem 1.1]{BHP13} relies on one-dimensional Monge-Kantorovich duality which works on the product of $T$ copies of $\R_+$ (i.e. \cite[Proposition 2.1]{BHP13}), one may expect to prove Proposition~\ref{thm:duality} via multi-dimensional Monge-Kantorovich duality which works on the product of $T$ copies of $\R^d_+$. While such a duality does exist (e.g. \cite[Theorem 2.14]{Kellerer84}), applying it requires the knowledge of the joint distribution of $(S^1_t,\dots,S^d_t)$ for each $t=1,\dots,T$. This is not practically feasible, as vanilla calls only specify the distribution of $S^n_t$, for each $n$ and $t$.

As a result, in Proposition~\ref{thm:duality}, we still rely on the one-dimensional result \cite[Proposition 2.1]{BHP13}. By treating $\Omega=(\R^d_+)^T$ as the product of $(d\times T)$ copies of $\R_+$ (as in Remark~\ref{rem:Pi compact}), \cite[Proposition 2.1]{BHP13} is indeed applicable as the distribution $\mu^n_t$ of $S^n_t$, for each $n=1,\dots,d$ and $t=1,\dots,T$, is known. Note that it was first mentioned in \cite[Theorem 2.1]{HL13} that \cite[Theorem 1.1]{BHP13} could be generalized to higher dimensions.
\end{rem}

By the convexity of $c_i$, Proposition~\ref{thm:duality} extends to the general case where $I\neq\emptyset$.

\begin{thm}\label{thm:duality main}
Suppose $\psi_i$ is continuous and $|\psi_i|$ satisfies \eqref{Phi linear grow} for all $i\in I$. Then, for any upper semicontinuous function $\Phi:(\R^d_+)^T\mapsto\R$ satisfying \eqref{Phi linear grow}, we have $D(\Phi)=P(\Phi)$, with $D$ and $P$ defined as in \eqref{def:D} and \eqref{def:P}. Moreover, if $\cQ_{\cS,I}\neq\emptyset$, the supremum in \eqref{def:P} is attained at some $\Q^*\in\cQ_{\cS,I}$. 
\end{thm}

\begin{proof}
Observe from \eqref{def:D} and Proposition~\ref{thm:duality} that 
\begin{equation}\label{D to D_empty}
\begin{split}
D(\Phi) &=\inf_{\eta\in\mathcal{R}^I} D_\emptyset\left(\Phi-\sum_{i\in I}(\eta_i\psi_i-c_i(\eta_i))\right)\\
&= \inf_{\eta\in\mathcal{R}^I}\sup_{\Q\in\Pi}\E^\Q\left[\Phi-\sum_{i\in I}(\eta_i\psi_i-c_i(\eta_i))-A_T^\Q\right].
\end{split}
\end{equation}
Consider the function $f(\Q,\eta):=\E^\Q[\Phi-\sum_{i\in I}(\eta_i\psi_i-c_i(\eta_i))-A_T^\Q]$ for $\Q\in\Pi$ and $\eta\in\mathcal{R}^I$. By the upper semicontinuity of $\Phi$, the continuity of $\psi_i$, and \eqref{Phi linear grow}, we may argue as in Proposition~\ref{thm:duality} (i) and (ii) that $f$ is upper semicontinuous in $\Q\in\Pi$. Moreover, by the convexity of $\eta\mapsto c_i(\eta)$ for all $i\in I$, $f$ is convex in $\eta\in\mathcal{R}^I$. Thus, we may apply Lemma \ref{lem:minimax} to \eqref{D to D_empty} and get
\[
\begin{split}
D(\Phi)&=\sup_{\Q\in\Pi}\inf_{\eta\in\mathcal{R}^I}\E^\Q\left[\Phi-\sum_{i\in I}(\eta_i\psi_i-c_i(\eta_i))-A_T^\Q\right]\\
&=\sup_{\Q\in\cQ_{\cS,I}}\Bigl\{\E^\Q[\Phi-A_T^\Q]-\cE^\Q_I\Bigr\}=P(\Phi).
\end{split}
\]
In view of the argument in Proposition~\ref{thm:duality} (iii) and the continuity of $\psi_i$, we obtain that $\Q\mapsto \E^\Q[\Phi-A_T^\Q]-\cE^\Q_I$ is upper semicontinuous on the compact set $\Pi$. Thus, the supremum in \eqref{def:P} is attained if $\cQ_{\cS,I}\neq\emptyset$.
\qed\end{proof}


\subsection{Connection to the model-free duality in \cite{ABPS13}}\label{sec:psi_i liquid}

Consider the case where every option $\psi_i$ can actually be liquidly traded at time $0$, just as vanilla calls. That is, for each $i\in I$, $c_i$ is linear (see Remark~\ref{rem:bid-ask spread}), and we let $p_i\in\R$ be the slope of $c_i$. By \eqref{cE^Q_I}, $\cE^\Q_I$ equals $0$ iff $\E^\Q[\psi_i]=p_i$ for all $i\in I$, and $\infty$ if otherwise. It follows from \eqref{Q_S,I} that
\begin{equation}\label{Q_S=Q_S,p_i}
\cQ_{\cS,I} =  \cQ_{\cS,(p_i)_{i\in I}} := \{\Q\in\cQ_\cS:\E^\Q[\psi_i]=p_i,\ \forall i\in I\}.
\end{equation}
Recall $\cM$ defined in \eqref{def:cM}. Let us also consider 
\begin{equation}\label{cM_I}
\cM_I:=\{\Q\in \cM: c'_i(0-)\le\E^\Q[\psi_i]\le c'_i(0+),\ \forall i\in I\}.
\end{equation}
Under current setting, it becomes
\begin{align}
\cM_I=\cM_{(p_i)_{i\in I}}  &:= \{\Q\in \cM: \E^\Q[\psi_i]=p_i,\ \forall i\in I\}.
\end{align}

\begin{coro}\label{coro:p_i constant}
For each $i\in I$, suppose $\psi_i$ is continuous, $|\psi_i|$ satisfies \eqref{Phi linear grow}, and $\psi_i$ can be traded liquidly at the price $p_i\in\R$. Let $\Phi:(\R^d_+)^T\mapsto\R$ be upper semicontinuous and satisfy \eqref{Phi linear grow}. 
\begin{itemize}
\item [(i)] We have 
\[
D(\Phi)=P(\Phi) = \sup_{\Q\in\cQ_{\cS,(P_i)_{i\in I}} }\E^\Q[\Phi-A^\Q_T].
\]
\item [(ii)] Furthermore, if there is no portfolio constraint, i.e. $\cS=\cH$, then
\[
D(\Phi)=P(\Phi) = \sup_{\Q\in\cM_{(p_i)_{i\in I}} }\E^\Q[\Phi].
\]
\end{itemize}
\end{coro} 

\begin{proof}
(i) simply follows from Theorem~\ref{thm:duality main} and \eqref{Q_S=Q_S,p_i}. For (ii), recalling from Remark~\ref{rem:S=H} that $\cS=\cH$ implies $\cQ_{\cS} = \cM$, we have $\cQ_{\cS,(p_i)_{i\in I}}=\cM_{(p_i)_{i\in I}}$. Then, part (i) just becomes the desired result.
\qed\end{proof}

\begin{rem}\label{rem:relate to [1]}
Corollary~\ref{coro:p_i constant} (ii) states that to find the superhedging price of $\Phi$, one needs to consider expectations of $\Phi$ under martingale measures which are consistent with market prices of both vanilla calls and other options $\{\psi_i\}_{i\in I}$. This in particular recovers \cite[Theorem 1.4]{ABPS13}, for the case where tradable options at time $0$ include at least vanilla calls with all maturities and strikes. 
\end{rem}


\subsection{Connection to convex risk measures}
Let $\mathcal{X}$ be the collection of measurable functions $\Phi:(\R^d_+)^T\mapsto\R$ satisfying the linear growth condition \eqref{Phi linear grow}. We say $\rho:\mathcal{X}\mapsto\R$ is a {\it convex risk measure} if for all $\Phi$, ${\Phi'\in\mathcal{X}}$, the following conditions hold:
\begin{itemize}
\item {\it Monotonicity:} If $\Phi\le\Phi'$, then $\rho(\Phi)\ge\rho(\Phi')$. 
\item {\it Translation Invariance:} If $m\in\R$, then $\rho(\Phi+m)=\rho(\Phi)-m$.
\item {\it Convexity:} If $0\le\lambda\le 1$, then $\rho(\lambda\Phi + (1-\lambda)\Phi')\le \lambda\rho(\Phi)+(1-\lambda)\rho(\Phi')$.
\end{itemize}
Consider  the {\it acceptance set}
\[
\begin{split}
\A_\cS&:=\{\Phi\in\mathcal{X}: u\in\cU_0,\ \eta\in\mathcal{R}^I,\ \hbox{and } \Delta\in\cS\ \hbox{such that } \\
&\hspace*{6cm}\Phi(x)+\Psi_{u,\eta,\Delta}(x)\ge 0, \forall x\in(\R^d)^T \}.
\end{split}
\]
Then, 
define the function $\rho_\cS:\mathcal{X}\mapsto\R$ by 
\[
\rho_\cS(\Phi):=\inf\{m\in\R:m+\Phi\in\A_\cS\}=D(-\Phi).
\]
\begin{prop}\label{prop:risk measure}
If $\cQ_{\cS,I}\neq\emptyset$, then $\rho_\cS$ is a convex risk measure, and admits the dual formulation
\begin{equation}\label{dual for rho}
\rho_\cS(\Phi) = \sup_{\Q\in\Pi}\left(\E^\Q[-\Phi]-\alpha^*(\Q)\right),
\end{equation}
where the {\it penalty function} $\alpha^*$ is given by
\[
\alpha^*(\Q):=
\begin{cases}
\E^\Q[A^\Q_T] +\cE^\Q_I\ \ &\hbox{if}\ \Q\in\cQ_{\cS,I},\\
\infty,\ \ &\hbox{otherwise}.
\end{cases}
\]
Moreover, for any $\alpha:\Pi\mapsto\R\cup\{\infty\}$ such that \eqref{dual for rho} holds (with $\alpha^*$ replaced by $\alpha$), we have $\alpha^*(\Q)\le \alpha(\Q)$ for all $\Q\in\Pi$.
\end{prop}

\begin{proof}
Monotonicity and translation invariance can be easily verified, while the convexity of $\rho_\cS$ follows from the convexity of $\cU_0$, $\mathcal{R}^I$ and $\cS$.  
Now, the duality \eqref{dual for rho} is a direct consequence of Theorem~\ref{thm:duality main}. Since $\cQ_{\cS,I}\neq\emptyset$, \eqref{dual for rho} shows that $\rho_\cS$ is real-valued, and thus a convex risk measure. To show that $\alpha^*$ is the minimal penalty function, observe that for any $\alpha:\Pi\mapsto\R\cup\{\infty\}$ satisfying \eqref{dual for rho}, we have 
\begin{equation}\label{alpha}
\begin{split}
\alpha(\Q)\ge \sup_{\Phi\in\mathcal{X}}\left(\E^\Q[-\Phi]-\rho_{\cS}(\Phi)\right)&\ge \sup_{\Phi\in\A_\cS}\left(\E^\Q[-\Phi]-\rho_{\cS}(\Phi)\right)\\
&\ge \sup_{\Phi\in\A_\cS}\E^\Q[-\Phi],\ \forall \Q\in\Pi.
\end{split}
\end{equation}
By Lemma~\ref{lem:S->S_c} and \eqref{cE^Q_I}, 
\[
\alpha^*(\Q) = \sup_{}\left\{\E^\Q\left[(\Delta\bcdot S)_T + \sum_{i\in I}(\eta_i\psi_i-c_i(\eta_i))\right] : \Delta\in\cS^\infty,\eta\in\mathcal{R}^I\right\}\]
 for all $\Q\in\Pi$.
Since $-(\Delta\bcdot S)_T - \sum_{i\in I}(\eta_i\psi_i-c_i(\eta_i))\in \A_\cS$ for all $\Delta\in\cS^\infty$ and $\eta\in\mathcal{R}^I$, we conclude from \eqref{alpha} that $\alpha^*(\Q)\le\alpha(\Q)$.
\qed\end{proof}

\begin{rem}
Proposition~\ref{prop:risk measure} generalizes Proposition 16 and Theorem 17 in \cite{FS02} to a model-independent setting. Note that a no-arbitrage condition (under a given physical measure $\P$) is imposed in \cite{FS02}. Here, we require only $\cQ_{\cS,I}\neq\emptyset$, which is weaker than the model-independent no-arbitrage condition; see Section~\ref{sec:NA} for details.   
\end{rem}


\section{Fundamental theorem of asset pricing via duality}\label{sec:NA}
Following the formulation in \cite{ABPS13}, we introduce the notion of arbitrage in the strong  pathwise sense:
\begin{defn}[model-independent arbitrage]\label{defn:NA}
We say there is model-independent arbitrage under the constraint $\cS$, if there exist $u\in\cU_0$, $\eta\in\mathcal{R}^I$, and $\Delta\in\cS$ such that
\begin{equation*}\label{def:arbitrage}
\sum_{t=1}^T\sum_{n=1}^du^n_t(x^n_t)+\sum_{i\in I}(\eta_i\psi_i(x)-c_i(\eta_i))+(\Delta\bcdot x)_T>0,\ \ \ \hbox{for all}\ x\in(\R^d_+)^T.
\end{equation*}
\end{defn}

\begin{rem}
It is immediate from the above definition that if model-independent arbitrage exists, then it is arbitrage under any probability measure $\P$ defined on $\Omega$. 

Note that instead of using the pathwise formulation, the authors in \cite{BN15} introduce a weaker notion of arbitrage under model uncertainty via quasi-sure analysis. They include more strategies in the definition of arbitrage, and provide different characterization of no-arbitrage condition and superhedging duality. We, however, will not pursue this direction in this paper.
\end{rem}

Consider the following collection of measure
\begin{align*}
\cP_\cS&:=\{\Q\in\Pi: \{(\Delta\bcdot S)_t\}_{t=0}^T\ \hbox{is a local}\ \Q\hbox{-supermartingale, for all}\ \Delta\in\cS\}\label{def:P_S}.
\end{align*}

\begin{rem}[$\cM$ and $\cP_\cS$]\label{rem:cM=cP}
By definition, we see that $\cM\subset\cP_\cS$. Given $\alpha>0$, if 
$\bar{\alpha}:=\{\Delta^n_t\equiv \alpha\}_{t,n}$ and ${-\bar\alpha}:=\{\Delta^n_t\equiv -\alpha \}_{t,n}$  
both belong to $\cS$, then $\cP_\cS=\cM$. Indeed, given $\Q\in\cP_\cS$, since $\bar{\alpha}$, $-\bar{\alpha}\in\cS^\infty$, $\alpha S_t =(\bar{\alpha}\bcdot S)_t$ and $-\alpha S_t=(-\bar{\alpha}\bcdot S)_t$ are both supermartingales under $\Q$ . We thus conclude that $\Q\in\cM$. 
\end{rem}

The following lemma provides with a characterization of $\cP_\cS$.
\begin{lem}\label{lem:P_S characterization}
Fix $\Q\in\Pi$. Then, $\Q\in\cP_\cS \iff A^\Q_T=0$ $\Q$-a.s.
\end{lem}

\begin{proof}
This is a consequence of \cite[Proposition 9.6]{FS-book-11} and Lemma \ref{lem:local supermart.}. 
\qed\end{proof}

\begin{rem} [$\cP_\cS$ and $\cQ_\cS$]\label{rem:cP=cQ}
Lemma~\ref{lem:P_S characterization} in particular implies that $\cP_\cS\subseteq\cQ_\cS$.
Observe that if $\cS^\infty$ is composed of all nonnegative bounded trading strategies in $\mathcal{H}$, then $\cP_\cS=\cQ_\cS$. Indeed, for any $\Q\in\cQ_\cS$, we see from \eqref{E[A_t]} that $A^\Q_T=0$ $\Q$-a.s. Then $\Q\in\cP_\cS$ by Lemma~\ref{lem:P_S characterization}.
\end{rem}

To state an equivalent condition for no-arbitrage, we consider 
\begin{equation*}\label{def:P_S,I}
\cP_{\cS,I} := \{\Q\in\cP_\cS : c'_i(0-)\le \E^\Q[\psi_i]\le c'_i(0+),\ \hbox{for all}\ i\in I\}.
\end{equation*}
Recall $\cE^\Q_I$ in \eqref{cE^Q_I}. It is easy to verify the following characterization for $\cE^\Q_I=0$. 

\begin{lem}\label{lem:cE^Q_I=0}
Given $\Q\in\Pi$, we have $c'_i(0-)\le \E^\Q[\psi_i]\le c'_i(0+)\ \hbox{for all $i\in I$}$ if and ony if $\cE^\Q_I= 0.$
\end{lem}

To derive a model-independent FTAP, we need the following lemma.

\begin{lem}\label{lem:infE[A_T]+cE}
Suppose $\psi_i$ is continuous and $|\psi_i|$ satisfies \eqref{Phi linear grow} for all $i\in I$. Then, 
\[
\cP_{\cS,I}=\emptyset\quad \implies\quad \inf_{\Q\in\Pi}\{\E^\Q[A^\Q_T]+\cE^\Q_I\}>0. 
\]
\end{lem}

\begin{proof}
Assume to the contrary that $\inf_{\Q\in\Pi}\{\E^\Q[A^\Q_T]+\cE^\Q_I\}=0$. For any $\eps>0$, there exists $\Q_\eps\in\Pi$ such that $0\le \E^{\Q_\eps}[A^{\Q_\eps}_T]+\cE^{\Q_\eps}_I<\eps$. Since $\Pi$ is weakly compact (Remark~\ref{rem:Pi compact}), $\Q_\eps$ must converge weakly to some $\Q^*\in\Pi$. For each $\Delta\in\cS^\infty_c$, we can argue as in Proposition~\ref{thm:duality} (iii) to show that $\Q\mapsto\E^\Q[(\Delta\bcdot x)_T]$ is continuous on $\Pi$ under the topology of weak convergence. Also, for each $i\in I$, since $\psi_i$ is continuous and $|\psi_i|$ satisfies \eqref{Phi linear grow}, we may argue as in Proposition~\ref{thm:duality} (ii) to show that $\Q\mapsto\E^\Q[\psi_i]$ is continuous on $\Pi$. Now, by using Lemma~\ref{lem:S->S_c},
\begin{align*}
0&=\lim_{\eps\to 0}\E^{\Q_\eps}[A^{\Q_\eps}_T]+\cE^{\Q_\eps}_I
 = \lim_{\eps\to 0}\sup_{\Delta\in \cS^\infty_c} \E^{\Q_\eps} [(\Delta\bcdot S)_T]+\sup_{\eta\in\mathcal{R}^I} \sum_{i\in I}(\eta_i \E^{\Q_\eps}[\psi_i]-c_i(\eta_i))\nonumber\\
\nonumber&\ge   \sup_{\Delta\in \cS^\infty_c} \lim_{\eps\to 0}\E^{\Q_\eps} [(\Delta\bcdot S)_T]+\sup_{\eta\in\mathcal{R}^I}\lim_{\eps\to 0} \sum_{i\in I}(\eta_i \E^{\Q_\eps}[\psi_i]-c_i(\eta_i))\\
&=\sup_{\Delta\in \cS^\infty_c} \E^{\Q^*} [(\Delta\bcdot S)_T] +\sup_{\eta\in\mathcal{R}^I} \sum_{i\in I}(\eta_i \E^{\Q^*}[\psi_i]-c_i(\eta_i))= \E^{\Q^*}[A^{\Q^*}_T]+\cE^{\Q^*}_I,
\end{align*}
Thus, $\E^{\Q^*}[A^{\Q^*}_T]=\cE^{\Q^*}_I=0$. By Lemmas~\ref{lem:P_S characterization} and \ref{lem:cE^Q_I=0}, we have $\Q^*\in\cP_{\cS,I}$, contradicting $\cP_{\cS,I}=\emptyset$.
\qed\end{proof}

Now, we are ready to present the main result of this section.  

\begin{thm}\label{thm:NA}
Suppose $\psi_i$ is continuous and $|\psi_i|$ satisfies \eqref{Phi linear grow} for all $i\in I$. Then, there is no model-independent arbitrage under constraint $\cS$ if and only if $\cP_{\cS,I}\neq\emptyset$.
\end{thm}

\begin{proof}
To prove ``$\Longleftarrow$'', suppose there is model-independent arbitrage. That is, there exist $u\in\cU_0$, $\eta\in\mathcal{R}^I$, and $\Delta\in\cS$ such that 
\[
\sum_{t=1}^T\sum_{n=1}^d u^n_t(x^n_t) + \sum_{i\in I}(\eta_i\psi_i(x)-c_i(\eta_i)) +(\Delta\bcdot x)_T >0\ \ \ \hbox{for all}\ x\in\R^T_+.
\]
It follows that for any $\Q\in\cQ_{\cS,I}$,
\[
\sum_{t=1}^T\sum_{n=1}^d u^n_t(S^n_t) + \sum_{i\in I}(\eta_i\psi_i(S)-c_i(\eta_i))+(\Delta\bcdot S)_T -A^\Q_T > -A^\Q_T\ \ \ \Q\hbox{-a.s.}
\]
By taking expectation on both sides, we obtain from Lemmas~\ref{lem:supermart.} that
\[
\cE^\Q_I\ge \sum_{i\in I}(\eta_i\E^\Q[\psi_i]-c_i(\eta_i))  > -\E^\Q[A^\Q_T],\ \ \ \hbox{for all}\ \Q\in\cQ_{\cS,I}.
\]
If $\Q\in\cP_{\cS,I}$, then the above inequality becomes $\cE^\Q_I>0$, thanks to Lemma~\ref{lem:P_S characterization}. However, in view of Lemma \ref{lem:cE^Q_I=0}, this implies $\E^\Q[\psi_i]\notin [c'_i(0-),c'_i(0+)]$ for some $i\in I$ and thus $\Q\notin\cP_{\cS,I}$, a contradiction. Hence, we conclude that $\cP_{\cS,I}=\emptyset$.\\
To show ``$\Longrightarrow$'', we assume to the contrary that $\cP_{\cS,I}=\emptyset$, and intend to find model-independent arbitrage. By Theorem~\ref{thm:duality main} and Lemma~\ref{lem:infE[A_T]+cE}, we have
\[
D(0)=\sup_{\Q\in\cQ_{\cS,I}}\{\E^\Q[-A_T^\Q]-\cE^\Q_I\}=-\inf_{\Q\in\cQ_{\cS,I}}\{\E^\Q[A_T^\Q]+\cE^\Q_I\}<0,
\]
which already induces model-independent arbitrage.
\qed\end{proof}

Let us recall the set-up in Section~\ref{sec:psi_i liquid}: every option $\psi_i$ can actually be liquidly traded at time $0$; that is, for each $i\in I$, we have $c'_i(\eta)$ being a constant $p_i\in\R$. Hence,
\begin{equation*}\label{P_S=P_S,p_i}
\cP_{\cS,I} = \cP_{\cS,(p_i)_{i\in I}}:= \{\Q\in\cP_\cS:\E^\Q[\psi_i]=p_i,\ \forall i\in I\}.
\end{equation*}

\begin{coro}\label{coro:NA p_i constant}
Suppose $\psi_i$ is continuous, $|\psi_i|$ satisfies \eqref{Phi linear grow}, and $\psi_i$ can be liquidly traded at the price $p_i\in\R$, for all $i\in I$.
\begin{enumerate}
\item [(i)] No model-independent arbitrage under constraint $\cS$ if and only if $\cP_{\cS,(p_i)_{i\in I}} \neq\emptyset$. 
\item [(ii)] Furthermore, suppose there is no portfolio constraint, i.e. $\cS=\cH$. Then, there is no model-independent arbitrage if and only if $\cM_{(p_i)_{i\in I}}\neq\emptyset$.
\end{enumerate}
\end{coro}

\begin{rem}
Corollary~\ref{coro:NA p_i constant} (ii) recovers \cite[Theorem 1.3]{ABPS13}, for the case where tradable options at time $0$ include at least vanilla calls of all maturities and strikes. 
\end{rem}

\begin{rem}\label{rem:stone-cech}
Among different model-independent versions of the fundamental theorem of asset pricing (FTAP), Theorem~\ref{thm:NA} and \cite[Theorem 1.3]{ABPS13} are unique in their ability to accommodate general collections of tradable options. While our framework deals with a wide range of tradable options beyond liquidly traded vanilla calls, \cite{ABPS13} does not even assume that vanilla calls have to be tradable. On the other hand, while \cite{ABPS13} implicitly assume that any tradable option is traded liquidly, we allow for less liquid options by taking into account their limit order books.

Also note that our method differs largely from that in \cite{ABPS13}. Techniques in functional analysis, which involves the use of Hahn-Banach theorem, are used to establish \cite[Theorem 1.3]{ABPS13}; see \cite[Proposition 2.3]{ABPS13}. In our case, we first derive a superhedging duality in Theorem~\ref{thm:duality main} via optimal transport. Leveraging on this duality, we obtain the desired FTAP in Theorem~\ref{thm:NA}.
\end{rem}

\subsection{Comparison with the classical theory}\label{sec:compare classical}
 In the classical theory, a physical measure $\P$ on $(\Omega, \cF_T)$ is a priori given. We say there is no-arbitrage under $\P$ with the constraint $\cS$ if, for any $\Delta\in\cS$, $(\Delta\bcdot S)_T\ge 0$ $\P$-a.s. implies $(\Delta\bcdot S)_T=0$ $\P$-a.s.

Consider the positive cone $\mathcal{K}:=\{\lambda\Delta \mid \Delta\in\cS,\ \lambda\ge 0\}$ generated by $\cS$. For all $t=1,\dots,T$, we define $\mathcal{S}_t := \{\Delta_t \mid \Delta\in\mathcal{S}\}$, $\mathcal{K}_t:=\{\Delta_t \mid \Delta\in\mathcal{K}\}$, and introduce
\begin{align*}
N_t &:=\{\eta\in  L^0(\Omega,\cF_{t-1},\P;\R^d) : \eta\cdot (S_t-S_{t-1}) =0\ \P\hbox{-a.s.}\},\\
N^\perp_t &:=\{\xi\in  L^0(\Omega,\cF_{t-1},\P;\R^d) : \xi\cdot \eta =0\ \P\hbox{-a.s. for all }\eta\in N_t\}.
\end{align*}
By \cite[Lemma 1.66]{FS-book-11}, every $\xi\in L^0(\Omega,\cF_{t-1},\P;\R^d)$ has a unique decomposition $\xi = \eta + \xi^\perp$, with $\eta\in N_t$ and $\xi^\perp\in N_t^\perp$.
We denote by $\hat{\mathcal{S}}_t$ and $\hat{\mathcal{K}}_t$ the closures of $\mathcal{S}_t$ and $\mathcal{K}_t$, respectively, in $L^0(\Omega,\cF_{t-1},\P;\R^d)$. The following characterization of no-arbitrage under $\P$ is taken from \cite[Theorem 9.9]{FS-book-11}.

\begin{prop}\label{thm:NA no option}
Suppose that for all $t=1,\dots,T$,
\begin{equation}\label{closedness}
\cS_t  = \hat{\mathcal{S}}_t,\quad  \hat{\mathcal{K}}_t\cap L^\infty(\Omega,\cF_{t-1},\P;\R^d)\subseteq \mathcal{K}_t,\quad\hbox{and}\quad \xi^\perp\in \cS_t\ \hbox{for any}\  \xi\in \cS_t. 
\end{equation}
Then, there is no arbitrage under $\P$ with the constraint $\cS$ if and only if
\[
\begin{split}
\cP_\cS(\P)& := \bigl\{\Q\approx\P : S_t\in L^1(\Q)\ \hbox{and}\ \{(\Delta\bcdot S)_t\}_{t=0}^T\ \hbox{is a local}\ \Q\hbox{-supermartingale},\\
&\hspace*{8.5cm} \forall\Delta\in\cS\bigr\}\neq\emptyset.
\end{split}
\]
\end{prop}

Theorem~\ref{thm:NA} can be viewed as a generalization of Proposition~\ref{thm:NA no option} to a model-independent setting. There is, however, a notable discrepancy: the closedness assumption \eqref{closedness} is no longer needed in Theorem~\ref{thm:NA}. In the following, we provide a detailed illustration of this discrepancy in a simple example.

A typical example showing that condition \eqref{closedness} is indispensable for Proposition~\ref{thm:NA no option} is a one-period model containing two risky assets $(S^1,S^2)$, with the collection of constrained strategies
\begin{equation*}\label{cS in eg}
\cS:=\{(\Delta^1,\Delta^2)\in\R^2\mid (\Delta^1)^2 + (\Delta^2-1)^2\le 1\}.
\end{equation*}
One easily sees that \eqref{closedness} is not satisfied, as $\hat{\mathcal{K}}_1\cap L^\infty(\Omega,\cF_0,\P)=\bar{\mathcal{K}}=\{\Delta^2\ge 0\}$ is not contained in $\mathcal{K}_1=\mathcal{K}=\{(0,0)\}\cup\{\Delta^2>0\}$. At time $0$, suppose $(S^1_0,S^2_0)=(1,1)$, and we obtain, by analyzing market data, that a reasonable pricing measure $\Q$ should be such that
\begin{equation}\label{S_1, S_2 distribution}
S^1_1\ \hbox{is uniformly distributed on}\ [1,2],\ \hbox{and}\ S^2_1\ \hbox{is concentrated solely on}\ \{0\}.
\end{equation}

Under the classical framework, the physical measure $\P$ should satisfy \eqref{S_1, S_2 distribution}, so that any pricing measure $\Q\approx\P$ admits the same property. Given $\Delta\in\cS$, it can be checked that if $(\Delta\cdot S)_T = \Delta^1(S^1_1-1)-\Delta^2 \ge 0$ $\P$-a.s., then $\Delta^1=\Delta^2=0$. There is therefore no arbitrage under $\P$ with the constraint $\cS$. However, as observed in \cite[Example 2.1]{BZ14}, $\cP_\cS(\P)$ is empty. Indeed, given $\Q\approx\P$, since $\E^\Q[S^1_1-1]>0$, by taking $\Delta\in\cS$ with $\Delta^2/\Delta^1<\E^\Q[S^1_1-1]$, one gets $\E^\Q[(\Delta\cdot S)_T]=\Delta^1\E^\Q[S^1_1-1]-\Delta^2>0$. 

Under our model-independent framework, \eqref{S_1, S_2 distribution} is reflected through market prices of vanilla calls $(S^1_1-K)^+$ and $(S^2_1-K)^+$, for all $K\ge 0$. That is, $\Pi$ in \eqref{Pi 2} is the collection $\{\Q\in\cP(\Omega):\eqref{S_1, S_2 distribution}\ \hbox{is satisfied}\}$. 
Given $\Q\in\Pi$, since $\E^\Q[S^1_1]=3/2$, by taking $\Delta\in\cS$ with $\Delta^2/\Delta^1<1/2$, one gets $\E^\Q[(\Delta\cdot S)_T]=\Delta^1\E^\Q[S^1_1-1]-\Delta^2>0$. This shows that $\cP_\cS=\emptyset$. Note that this does not violate Theorem~\ref{thm:NA}, as there {\it is} model-independent arbitrage. To see this, consider trading dynamically with $\Delta\in\cS$ satisfying $\Delta^2<\eps\Delta^1$ for some $\eps\in (0,1)$, and holding a static position $u^1_1$ given by ${-\Delta^1 (S^1_1-\eps)^+ +(1+\eps)\Delta^1\in\mathcal{C}}$. Observe that the 
initial wealth required is 
\[
\int_{\R_+} u^1_1(x) d\mu^{1}_1(x) = -\Delta^1(3/2-\eps)+(1+\eps)\Delta^1=(2\eps-1/2)\Delta^1,
\]
while the terminal wealth is always strictly positive: for any $(S^1_1,S^2_1)\in\R^2_+$,
\begin{equation}\label{arbitrage eg}
u^1_1 + (\Delta\cdot S)_1 = 
\begin{cases}
2\eps\Delta^1 + \Delta^2(S^2_1-1)>\eps\Delta^1-\Delta^2 >0,\ \ &\hbox{if}\ S^1_1\ge \eps,\\
(\eps\Delta^1-\Delta^2) + (\Delta^1S^{1}_1 +\Delta^2S^2_1)> 0,\ \ &\hbox{if}\ S^2_1<\eps.
\end{cases}
\end{equation}
By taking $\eps\in(0,1/4]$, we have the required initial wealth no greater than $0$, and thus obtain model-independent arbitrage.

\begin{rem}
In the model-independent setting, we may hold static positions in
 vanilla calls $(S^i_1 - K)^+$ for all $i\in\{1,2\}$ and $K>0$, in addition to trading $S=(S^1,S^2)$ dynamically. This additional flexibility, unavailable under the classical framework, allows us to construct the arbitrage in \eqref{arbitrage eg}.  

It is of interest to see if \eqref{closedness} can be relaxed in the classical case, when enough tradable options are available at time $0$. Recently, with additional tradable options, \cite{BZ14} obtained a result similar to Proposition~\ref{thm:NA no option} with a collection $\mathcal{P}$ of possible physical measures a priori given. However, since their method allows for only {\it finitely many} tradable options, a closedness assumption similar to \eqref{closedness} is still assumed.
\end{rem}

\subsection{Optimal arbitrage under a model-independent framework} In view of Theorem~\ref{thm:duality main} and Proposition \ref{prop:risk measure}, the problems of superhedging and risk-measuring are well-defined as long as $\cQ_{\cS,I}\neq\emptyset$, which is weaker than the no-arbitrage condition $\cP_{\cS,I}\neq\emptyset$. It is therefore of interest to provide characterizations for the condition $\cQ_{\cS,I}\neq\emptyset$. 

\begin{defn}
Consider
\begin{equation}\label{OA}
G_{\cS,I}:=\sup\{ a\in\R : u\in\cU_0, \eta\in\mathcal{R}^I, \hbox{and } \Delta\in\cS \hbox{such that } \Psi_{u,\eta,\Delta}(x)>a, \forall x\in(\R^d_+)^T\}.
\end{equation}
By definition, $G_{\cS,I}\ge 0$. If $G_{\cS,I}>0$, we say it is the (model-independent) {\it optimal arbitrage profit}.
\end{defn}
The notion of optimal arbitrage goes back to \cite{FK10OA}, where the authors studied the highest return one can achieve relative to the market capitalization in a diffusion setting. Generalization to semimartingale models and model uncertainty settings have been done in \cite{CT13} and \cite{FK11}, respectively. Our definition above is similar to the formulation in \cite[Section 3]{CT13}. It is straightforward from the definitions of $D(0)$ and $G_{\cS,I}$ that
\[
G_{\cS,I} = -D(0) =\inf_{\Q\in\cQ_{\cS,I}}\{\E^\Q[A^\Q_T] + \cE^\Q_I\}.
\]
This immediately yields the following result.

\begin{prop}\label{thm:NUP} 
$\ \ (i)$\ \ $G_{\cS,I}=0$ $\iff$ $\cP_{\cS,I}\neq\emptyset$.\quad $(ii)$\ \ $G_{\cS,I}<\infty$ $\iff$ $\cQ_{\cS,I}\neq\emptyset$.
\end{prop}


\section{Examples}\label{sec:examples}
In this section, we will provide several concrete examples of the collection $\cS$ of constrained trading strategies. An example which illustrates the effect of an additional tradable, yet less lquid, option will also be given. It will be convenient to keep in mind the relation
$\cM\subseteq\cP_\cS\subseteq\cQ_\cS\subseteq \Pi$,
obtained from Remarks~\ref{rem:cM=cP} and \ref{rem:cP=cQ}. Let us start with analyzing $\cQ_\cS$ further.

\begin{prop}\label{prop:cP=cQ}
Let $\cS^\infty$ contain all nonnegative bounded trading strategies in $\mathcal{H}$. 
\begin{itemize}
\item [(i)] For any $\Q\in\cQ_\cS$, $\{S_t\}_{t=0}^T$ is a $\Q$-supermartingale.
\item [(ii)] Furthermore, if trading strategies in $\cS^\infty$ are uniformly bounded from below, i.e.
 \[
\sup_{\Delta\in\cS}\sup_{x\in(\R^T_+)^d}|\Delta^-(x)|\le C\quad  \hbox{for some}\ C>0,
\] 
then $\cQ_\cS = \{\Q\in\Pi:\{S_t\}_{t=0}^T\ \hbox{is a}\ \Q\hbox{-supermartingale}\}$.
\end{itemize}
\end{prop}

\begin{proof}
(i) Given $\Q\in\cQ_\cS$, if $\{S_t\}_{t=0}^T$ is not a $\Q$-supermartingale, there must exist  $n^*\in\{1,\dots,d\}$ and $t^*\in\{0,\dots,T-1\}$ such that $\Q(\E^\Q[S^{n^*}_{t^*+1}\mid\cF_{t^*}]-S^{n^*}_{t^*}>0)>0$. We then deduce from \eqref{increment A^Q_t} that $\Q(A_{t^*+1}^\Q-A_{t^*}^\Q=\infty)>0$. This implies $\E^\Q[A_T^\Q]=\infty$, a contradiction to $\Q\in\cQ_\cS$.  \\
(ii) Let $\Q\in\Pi$ be such that $\{S_t\}_{t=0}^T$ is a $\Q$-supermartingale. It can be easily checked that $\{(\Delta\bcdot S)_t\}_{t=0}^T$ is a $\Q$-supermartingale, for any nonnegative bounded trading strategies $\Delta\in\mathcal{H}$.  
By \eqref{E[A_t]}, 
\begin{align*}
\E^\Q[A^\Q_T]&=\sup_{\Delta\in\cS^\infty} \left(\E^\Q[(\Delta^+\bcdot S)_T]-\E^\Q[(\Delta^-\bcdot S)_T]\right)\\
&\le \sup_{\Delta\in\cS^\infty}\E^\Q[|(\Delta^-\bcdot S)_T|]\le 2C\sum_{t=1}^T\sum_{n=1}^d(\E^\Q[S^n_{t+1}]+\E^\Q[S^n_t])<+\infty,
\end{align*}
which implies that $\Q\in\cQ_\cS$.
\qed\end{proof}

\begin{eg}[Shortselling constraint] \label{eg:shortselling}
Given $c^{n}_t\ge 0$ for each $t=0,\dots,T-1$  and $n=1,\dots, d$, we see from Remark~\ref{rem:adapted convex} (with $K_t=\prod_n[-c^{n}_t,\infty)$ for all $t$) that
\begin{equation*}
\cS:=\{\Delta\in\mathcal{H}: \Delta^n_t\ge -c^{n}_t,\ \forall t=0,\dots,T-1\text{ and } n=1,\dots,d\}.
\end{equation*}
satisfies Definition \ref{defn:cS}. By Proposition~\ref{prop:cP=cQ}, we have 
\[
\cQ_\cS=\{\Q\in\Pi : \{S_t\}_{t=0}^T\ \hbox{is a}\ \Q\hbox{-supermartingale}\}.
\]
Furthermore, if  $c^{n}_t>0$ for all $t$ and $n$, then $\cM=\cP_\cS$ by virtue of Remark \ref{rem:cM=cP}. If there exists $n\in\{1,\dots,d\}$ such that $c^{n}_t=0$ for all $t$, then $\{S^n_t\}_{t=0}^T$ is  a $\Q$-supermartingale for all $\Q\in\cP_\cS$. Thus, if $c^n_t=0$ for all $t$ and $n$, then $\cP_\cS=\cQ_\cS$.
\end{eg}

\begin{eg}[Relative-drawdown constraint]\label{eg:drawdown}
 Let $x^n_0>0$ for $n=1,\dots,d$. For each $x\in(\R^d_+)^T$, $t=1,\dots,T$ and $n\in\{1,\dots,d\}$, consider the running maximum $x^{*,n}_t$ given by $\max\{ x^n_0,x^n_1,\dots,x^n_t \}$. Then, define the relative drawdown process ${\{\tilde x_t : t=0,\dots,T\}}$ by $\tilde x_t := ( x^1_t/x^{*,1}_t , \dots, x^d_t/x^{*,d}_t ) $. For each $n=1,\dots,d$, take two continuous functions $a^n:[0,1]^d\mapsto(-\infty,0]$ and $b^n:[0,1]^d\mapsto [0,\infty)$, and introduce 
\begin{align*}
\cS &:=\{\Delta\in\mathcal{H}: a^n(\tilde x_t)\le \Delta^n_t(x_t) \le b^n(\tilde x_t),\ \forall t=0,\dots,T-1,\ n=1,\dots,d\}.\\
&=\{\Delta\in\mathcal{H}: \Delta_t\in K_t,\ \forall t=0,\dots,T-1\},\ \ \hbox{with}\ K_t:=\prod_{i=1}^d[a^n(\tilde x_t), b^n(\tilde x_t)].
\end{align*}
Since $K_t$ satisfies \eqref{lower semiconti.'}, Remark~\ref{rem:adapted convex} shows that $\cS$ satisfies Definition~\ref{defn:cS}. Thanks to Remark~\ref{rem:uniform bdd}, we have $\cQ_\cS=\Pi$. 
\end{eg}
 
\begin{eg}[Non-tradable assets]\label{eg:non-tradable}
Suppose certain risky assets are not tradable. In markets of electricity and foreign exchange rates, for example, people trade options written on the non-tradable underlying. More precisely, let $d'\in\{1,\dots,d\}$ and set
\[
\cS:=\{\Delta\in\cH :  \Delta_t^n\equiv0,\quad \hbox{for all}\ t=1,\dots,T\ \hbox{and}\ n=1,\dots, d'\}.
\]  
By a similar argument in Proposition~\ref{prop:cP=cQ}, one can show that
\[
\cP_{\cS}=\cQ_{\cS}=\{\Q\in\Pi: \{S_t^n\}_{t=1}^T \text{ is a $\Q$-martingale},\ \ \text{for all } n=d'+1,\dots,d \}.
\]
By Theorem \ref{thm:NA}, there is no model-independent arbitrage if and only if there exists $\Q\in\Pi$ under which all tradable assets are martingales. We can also modify this example by imposing additional constraint on the tradable assets satisfying Definition~\ref{defn:cS}. 
In this case, Theorem \ref{thm:NA} suggests that there is no arbitrage if and only if there is no arbitrage in the market consisting of tradable assets only.
\end{eg}

\begin{eg}[Less Liquid Option]\label{eg:illiquid option}
Consider a two-period model with one risky asset starting from $S_0=2$. We assume as in \cite[Section 4.2]{BHP13} that the marginal distributions for $S_1$ and $S_2$ are given by 
\[
\begin{split}
d\mu_1(x)=\frac121_{[1,3]}(x)dx,\; \quad \;d\mu_2(x)=\frac{x}{3}1_{[0,1]}(x)dx+\frac{1}{3}1_{[1,3]}(x)dx+\frac{4-x}{3}1_{[3,4]}(x)dx.
\end{split}
\] 
In addition to vanilla calls, we assume that a forward-start straddle with payoff $\psi(S)=|S_2-S_1|$ is also tradable at time $0$, whose unit price for trading $\eta$ units is given by ${p(\eta):=\infty1_{\{\eta> 1\}}+a1_{\{0\le \eta\le 1\}}+b1_{\{-1\le \eta<0\}}}$,
where $0\le b\le a$, and we take $0\cdot\infty=0$. We assume that the portfolio constraint $\cS$ satisfies $\cQ_{\cS}=\cM$. This readily covers the no-constraint case, as explained in Remark~\ref{rem:S=H}. Moreover, it also includes the shortselling constraint as in Example~\ref{eg:shortselling}. To see this, note from Example~\ref{eg:shortselling} that $\cQ_{\cS}=\{\Q\in\Pi: \{S_t\}_{t=0}^T\ \hbox{is a $\Q$-supermartingale}\}$. But since $\Q\in\Pi$ implies that $\E^\Q[S_1]=\E^\Q[S_2]=2$ (computed from $\mu_1$ and $\mu_2$), every $\Q\in\cQ_\cS$ is actually a martingale. We thus obtain $\cQ_\cS=\cM$.

We intend to price an exotic option with payoff $\Phi(x_1,x_2)=(x_2-x_1)^2$. Our goal is to see how using the additional option $\psi$ in static hedging affects the superhedging price of $\Phi$. 
First, for any $\Q\in\cM$, the martingale property of $S$ implies ${\E^\Q[(S_2-S_1)^2]}=\E^\Q[S_2^2]-\E^\Q[S_1^2]=\frac{1}{2}$, which is obtained solely from $\mu_1$ and $\mu_2$. Since $\cQ_\cS=\cM$, Proposition~\ref{thm:duality} gives $D_\emptyset(\Phi)=\frac{1}{2}$. On the other hand, $\cQ_\cS=\cM$ and $\cE^\Q_{\{\psi\}}<\infty$ for all $\Q\in\cM$ (see Remark~\ref{rem:cE^Q_I finite}) imply that $\cQ_{\cS,I}=\cM$. Theorem~\ref{thm:duality main} thus yields 
\[
\begin{split}
D(\Phi)&=\sup_{\Q\in\cM}(\E^\Q[(S_2-S_1)^2]-\cE_I^\Q)=\frac{1}{2}-\inf_{\Q\in\cM}\sup_{\eta\in[-1,1]}\eta(\E^\Q[|S_2-S_1|]-p(\eta))\\
&=\frac{1}{2}-\sup_{\eta\in[-1,1]}\eta(\inf_{\Q\in\cM}\E^\Q[|S_2-S_1|]-p(\eta)),
\end{split}
\]
where in the second line we used Lemma \ref{lem:minimax}. Recalling from \cite[Section 4.2]{BHP13} that $\inf_{\Q\in\cM}\E^\Q[|S_2-S_1|]=\frac{1}{3}$, 
we get $D(\Phi)=\frac{1}{2}-\max\left\{\left(\frac{1}{3}-a\right)^+,\left(b-\frac{1}{3}\right)^+\right\}$.
\end{eg}


\section{Bounded constraints without adapted convexity}\label{sec:bounded}

In this section, we extend the main results of this paper, Theorems \ref{thm:duality main} and \ref{thm:NA}, to a class of constraints which does not satisfy adapted convexity (Definition \ref{defn:cS} (ii)), but instead admits additional boundedness property. Motivations behind this include Gamma constraint, which will be discussed in Sections \ref{sec:Gamma}.

\begin{defn}\label{defn:cS_b}
$\mathcal{S}$ is a collection of trading strategies satisfying  Definition~\ref{defn:cS} (i) and (iii), while condition (ii) is replaced by the following:
\begin{itemize}
\item [(ii)$'$] (Boundedness) For any $\Delta\in\cS$,  $\exists\ c>0$ such that $|\Delta_t(x)|\le c$ for all $x\in(\R^d_+)^t$ and $t\in\{0,\dots,T-1\}$ (i.e. $\cS=\cS^\infty$).
\end{itemize}
\end{defn}
Under current setting, Lemma~\ref{lem:E[A_t]} does not hold anymore, and thus the upper variation process $A^\Q_t$ is no longer useful. We adjust the definitions of $\cQ_\cS$ and $\cP_\cS$ accordingly. 

\begin{defn}
For any $\Q\in\Pi$, we define
\begin{equation}\label{C^Q}
C^\Q:=\sup_{\Delta\in\cS}\E^\Q[(\Delta\bcdot S)_T].
\end{equation}
In analogy to $\cQ_\cS$ in Definition~\ref{def:Q_S} and the characterization of $\cP_\cS$ in Lemma~\ref{lem:P_S characterization}, we define 
 \begin{equation}\label{defn Q_S}
 \cQ'_\cS:=\{\Q\in\Pi\;:\; C^\Q<+\infty\}
 \quad\text{ and }\quad
 \cP'_{\cS}:=\{\Q\in\Pi\;:\; C^\Q=0\}.
 \end{equation}
\end{defn}
Recall from \eqref{defn cS^infty} that $\cS_{c}$ denotes the collection of $\Delta\in\cS$ with $\Delta_t:(\R^d_+)^t\mapsto\R^d$ continuous for all $t=1,\dots,T$. Using the arguments in Lemma~\ref{lem:S->S_c} gives
\begin{equation}\label{S->S_c C^Q}
C^\Q= \sup_{\Delta\in\cS_{c}}\E^\Q[(\Delta\bcdot S)_T],\ \ \ \forall \Q\in\Pi.
\end{equation}
By \eqref{C^Q}, \eqref{S->S_c C^Q}, and similar arguments in Proposition~\ref{thm:duality} (with $\E^\Q[A^{\Q}_T]$ replaced by $C^\Q$), we obtain:

\begin{prop}\label{prop:duality bdd}
Suppose $\cS$ satisfies Definition \ref{defn:cS_b}. Let $\Phi:(\R^d_+)^T\mapsto\R$ be a measurable function for which there exists $K>0$ such that \eqref{Phi linear grow} holds.
\begin{itemize}
\item [(i)] We have
\[
P'_\emptyset(\Phi):=\sup_{\Q\in\cQ'_\cS}\E^\Q[\Phi]-C^\Q \le D_\emptyset(\Phi).
\]
\item [(ii)] Furthermore, if $\Phi$ is upper semicontinuous, then $P'_\emptyset(\Phi)=D_\emptyset(\Phi)$.
\item [(iii)]  If $\Phi$ is upper semicontinuous and $\cQ'_\cS\neq\emptyset$, there exists $\Q^*\in\cQ'_\cS$ such that 
$P'_\emptyset(\Phi)=\E^{\Q^*}[\Phi]-C^{\Q^*}$.
\end{itemize}
\end{prop}

\begin{rem}
Under adapted convexity (Definition \ref{defn:cS} (ii)), Lemma \ref{lem:E[A_t]} asserts that $\E^\Q[A^\Q_T] = C^\Q$. This need not be true in general. For any $\Q\in\Pi$, we observe from Lemma \ref{lem:E[A_t]} that in general $\E^\Q[A^\Q_T] \ge \sup_{\Delta\in\cS^\infty} \E^\Q[(\Delta\bcdot S)_T] = C^\Q$.
This in particular implies that $P'_\emptyset(\Phi) \ge P_\emptyset(\Phi)$.
\end{rem}

We now include the collection of options $\{\psi_i\}_{i\in I}$ in the superhedging strategy. Recalling the definition of $\cE^\Q_I$ from \eqref{cE^Q_I}, we consider the collection of measures $\cQ'_{\cS,I}:=\{\Q\in\cQ'_\cS\;:\;\cE^\Q_I<\infty\}$,  
and define
\begin{equation}\label{def:P'}
P'(\Phi):=\sup_{\Q\in\cQ'_\cS}\E^\Q[\Phi]-C^\Q -\cE_I^\Q.
\end{equation}
The following result follows from a straightforward adjustment of Theorem \ref{thm:duality main}.

\begin{prop}\label{prop:duality bdd price impact}
Suppose $\cS$ satisfies Definition \ref{defn:cS_b}, $\psi_i$ is continuous and $|\psi_i|$ satisfies \eqref{Phi linear grow} for all $i\in I$. Then, for any upper semicontinuous function $\Phi:(\R^d_+)^T\mapsto\R$ satisfying \eqref{Phi linear grow}, we have $D(\Phi)=P'(\Phi)$, with $D$ and $P'$ defined as in \eqref{def:D} and \eqref{def:P'}. Moreover, if $\cQ'_{\cS,I}\neq\emptyset$, the supremum in \eqref{def:P'} is attained at some $\Q^*\in\cQ'_{\cS,I}$. 
\end{prop}

To derive the FTAP, we consider the collection of measures
\[
\cP'_{\cS,I}:=\{\Q\in\cP'_\cS\;:\; c'_i(0-)\le\E^\Q[\psi_i]\le c'_i(0+),\;\text{ for all } i\in I\}.
\]
By \eqref{S->S_c C^Q}, the same arguments in Lemma \ref{lem:infE[A_T]+cE} (with $\E^\Q[A^\Q_T]$ replaced by $C^\Q$) yields:
\begin{equation}\label{infC^Q}
\cP'_{\cS,I}=\emptyset\ \implies \inf_{\Q\in\cQ'_{\cS}} \{C^\Q+\cE^\Q_I\}>0.
\end{equation}
On strength of \eqref{infC^Q} and Proposition~\ref{prop:duality bdd price impact}, we may argue as in Theorem~\ref{thm:NA} and Proposition~\ref{thm:NUP}  (with $\E^\Q[A^\Q_T]$ replaced by $C^\Q$) to establish the following.
 
\begin{prop}\label{prop:NA bdd}
Suppose $\cS$ satisfies Definition~\ref{defn:cS_b}. Then, 
\begin{enumerate}
\item[(i)]  No model-independent arbitrage under the constraint $\cS$ if and only if $\cP'_{\cS,I}\neq\emptyset$.
\item[(ii)] Optimal arbitrage profit is finite under the constraint $\cS$ (i.e. $G_{\cS,I}<\infty$) if and only if $\cQ'_{\cS,I}\neq\emptyset$.
\end{enumerate}
\end{prop}


\subsection{Gamma constraint}\label{sec:Gamma}

Given $\Gamma=(\Gamma_1,\dots,\Gamma_d)\in\R^d_+$, we consider the collection of trading strategies
\[
\cS_\Gamma:=\{\Delta\in\mathcal{H}:|\Delta^n_t-\Delta^n_{t-1}|\le\Gamma_n,\ \forall t=0,\dots,T-1,\; n=1,\dots,d\},
\]
where we set $\Delta_{-1}\equiv 0\in\R^d$.
Observe that $\cS_\Gamma$ does not admit adapted convexity (Definition~\ref{defn:cS} (ii)). Indeed, consider $\Delta\equiv 0$ and $\Delta':=\{1_{\{t=0\}}\Gamma+1_{\{t>0\}} 2\Gamma\}_{t=0}^{T-1}$, both of which trivially lie in $\cS_\Gamma$. Given a fixed $s\in\{1,\dots,T-1\}$, the trading strategy $\tilde{\Delta}:= \{\Delta_t 1_{\{t<s\}} + \Delta'_t 1_{\{t\ge s\}}\}_{t=0}^{T-1}$ does not belong to $\cS_\Gamma$, as $\tilde{\Delta}_s - \tilde{\Delta}_{s-1}=2\Gamma$. The constrained collection $\cS_\Gamma$, instead, satisfies Definition \ref{defn:cS_b}. 

\begin{lem}\label{lem:Gamma bdd}
$\cS_\Gamma$ satisfies Definition \ref{defn:cS_b}.
\end{lem}

\begin{proof}
It is trivial that $0\in\cS_\Gamma$. For each $\Delta\in\cS_\Gamma$, since $\Delta_t=\sum_{j=0}^t(\Delta_j-\Delta_{j-1})$, we have $|\Delta_t|\le (t+1)|\Gamma|$, which shows that Definition \ref{defn:cS_b} (ii)$'$ is satisfied. It remains to prove Definition \ref{defn:cS} (iii).\\
In view of Remark \ref{rem:adapted convex}, it follows from Lusin's theorem that for any $\Q\in\Pi$ and $\eps>0$, there exist a closed set $D_\eps\subseteq\Omega$ and a sequence of continuous functions $\Delta^\eps(x)=\{\Delta^\eps_t(x_1,\dots,x_t)\}_{t=0}^{T-1}$ such that $\Q(D_\eps)>1-\eps$ and $\Delta=\Delta^\eps$ on $D_\eps$. That is, for all $t=1,\dots,T-1$, $\Delta_t$ is a continuous function when it is restricted to the domain $\hbox{proj}_{(\R^d_+)^t}{D_\eps}:=\{x\in(\R^d_+)^t:\exists\ y\in(\R^d_+)^{T-t}\ \hbox{such that}\ (x,y)\in D_\eps\}$. 
In the following, by induction over time $t$, we will construct a continuous strategy $\bar{\Delta}^\eps\in\cS_{\Gamma,c}$. At time $t=0$, $\bar\Delta_0^\eps:=\Delta_0$ is a constant in $\prod_{n=1}^d[-\Gamma_n,\Gamma_n]$, and therefore continuous. Fix $t\ge 1$. We assume that we have constructed continuous functions $\{\bar\Delta_s^\eps:(\R^d_+)^s\mapsto\R^d\}_{s=0}^{t-1}$ such that $\bar\Delta_s^\eps=\Delta_s$ on $\hbox{proj}_{(\R^d_+)^s}{D_\eps}$ and $|\bar\Delta_s^\eps-\bar\Delta_{s-1}^\eps|\le\Gamma$ on $(\R^d_+)^s\setminus\hbox{proj}_{(\R^d_+)^s}{D_\eps}$, for any $s<t$. By the continuity of $\bar\Delta_{t-1}^\eps$, the set-valued function defined by  
\[
K_t(x_1,\dots,x_t):=\left\{
\begin{array}{ll}
\{\Delta_t(x_1,\dots,x_t)\}&\ \  \text{on}\ \  \  \hbox{proj}_{(\R^d_+)^t}{D_\eps}\\
 \Delta_{t-1}(x_1,\dots,x_{t-1})+\prod_{n=1}^d[-\Gamma_n,\Gamma_n]&\ \  \text{on}\ \  \  (\R^d_+)^t\setminus\hbox{proj}_{(\R^d_+)^t}{D_\eps}
 \end{array}
 \right.
\]
satisfies \eqref{lower semiconti.'} and thus admits a continuous selection (\cite[Theorem 3.2$''$]{Michael56}); i.e. there is a continuous function $\bar\Delta_t^\eps:(\R^d_+)^t\mapsto\R^d$ such that, ${\bar\Delta_t^\eps(x_1,\dots,x_t)} \in K_t(x_1,\dots,x_t)$ for all $x\in\R^t_+$. Thus, we can construct $\bar\Delta^\eps\in\cS_{\Gamma,c}$, as required by Definition \ref{defn:cS} (iii).
\qed\end{proof}

\begin{prop}\label{prop:Q_S P_S Gamma}
$\cQ'_{\cS_\Gamma} = \Pi$ and $\cP'_{\cS_\Gamma} = \cM$.
\end{prop}

\begin{proof}
From the proof of Lemma~\ref{lem:Gamma bdd}, every $\Delta\in\cS_\Gamma$ is bounded by $c:=(T+1)|\Gamma|$. This gives $\cQ'_{\cS_\Gamma} = \Pi$, by Remarks~\ref{rem:uniform bdd}. For any $t\in\{0,\dots,T-1\}$ and  $A\in\cF_t$, observe that $\Delta^{(+)}=\{\Delta^+_s\}_{s=0}^{T-1} :=+\Gamma 1_A1_{\{s=t\}}$ and $\Delta^{(-)}=\{\Delta^-_s\}_{s=0}^{T-1} :=-\Gamma 1_A1_{\{s=t\}}$ both belong to $\cS_\Gamma$. Given $\Q\in\cP'_{\cS_\Gamma}$, the definition of $\cP'_{\cS_\Gamma}$ in \eqref{defn Q_S} implies that $\E^\Q[\Gamma 1_A (S_{t+1}-S_t)]=0$. This readily implies $\E^\Q[S_{t+1}\mid\cF_t] = S_t$, and thus $\Q\in\cM$.
\end{proof}

By Proposition~\ref{prop:Q_S P_S Gamma}, the following is a direct consequence of Proposition \ref{prop:NA bdd}. 

\begin{coro}\label{coro:NA Gamma}
$\mathcal{S}_\Gamma$ satisfies the following:
\begin{enumerate}
\item [(i)] There is no model-independent arbitrage under $\mathcal{S}_\Gamma$ if and only if $\cM_I\neq\emptyset$. 
\item [(ii)] Optimal arbitrage profit is finite under $\mathcal{S}_\Gamma$ (i.e. $G_{\cS_\Gamma,I}<\infty$) if and only if $\cE^\Q_I<\infty$ for some $\Q\in\Pi$. 
\end{enumerate}
\end{coro}

\begin{rem}
By Theorem~\ref{thm:NA}, Remark~\ref{rem:S=H}, and Corollary \ref{coro:NA Gamma}, we have equivalence between:
\begin{itemize}
\item [(i)] There is model-independent arbitrage with $\Delta\in\mathcal{H}$ (i.e. the no-constraint case).
\item [(ii)] There is model-independent arbitrage with $\Delta\in\cS_\Gamma$.
\end{itemize}
While these two arbitrage opportunities coexist, they are very different in terms of optimal arbitrage profit defined in \eqref{OA}. By Proposition~\ref{thm:NUP}, we see that ${G_{\mathcal{H}}<\infty}$ if and only if $\cQ_{\mathcal{H},I}=\{\Q\in\cM:\cE^\Q_I<\infty\}\neq\emptyset$, while $G_{\cS_{\Gamma}}<\infty$ if and only if ${\cQ'_{\cS_{\Gamma},I}=\{\Q\in\Pi:\cE^\Q_I<\infty\}\neq\emptyset}$.
\end{rem}


\appendix\normalsize

\section{An example related to definition \ref{defn:cS} (iii)}\label{sec:appendix}
In this appendix, we provide an example showing that if Definition \ref{defn:cS} (iii) is not satisfied, the duality in Proposition~\ref{thm:duality} may fail. Let $d=1$, $T=2$ and $x_0=1$.
Assume $\mu_1(dx)=\frac12\delta_{1}(dx)+\frac12\delta_{2}(dx)$ and $\mu_2(dx)=\delta_{2}(dx)$, where $\delta_x$ is the Dirac  measure at $x\in\R$. Thus, $\Pi=\{\Q\}$ with ${\Q(S_1=1,S_2=2)}=\Q(S_1=2,S_2=2)=\frac12$. Consider the collection of trading strategies
\[
\cS=\{\Delta=(\Delta_0,\Delta_1): \Delta_0\equiv 0,\ \Delta_1(x)=\alpha1_{\{x=1\}}(x)\ \hbox{for some}\ \alpha\in[0,1]\}.
\]
While $\cS$ trivially satisfies Definition \ref{defn:cS} (i) and (ii), Definition \ref{defn:cS} (iii) does not hold. To see this, note that $\cS_c^\infty=\{(0,0)\}$, and thus for any $\Delta\in\cS$ with $\alpha>0$, we have $\Q(\Delta\neq(0,0))= 1/2$. In order to superhedge the claim $\Phi(x_1,x_2)\equiv0$, we need to find $n, m\in\N$, $a, b_i, c_j\in\R$, $K_i^1, K_j^2\ge0$ and  $\Delta\in\cS$ such that for all $(x_1,x_2)\in\R^2_+$
\begin{equation*}
0\le a+\sum_{i=1}^nb_i(x_1-K_i^1)^++\sum_{j=1}^mc_j(x_2-K_j^2)^++\Delta_0(x_1-x_0)+\Delta_1(x_1)(x_2-x_1).
\end{equation*}
Since $\Delta_0\equiv0$ and $\Delta(x_1)=\alpha1_{\{x_1=1\}}$, the above inequality reduces to 
\begin{equation}
f_\alpha(x_1,x_2) := -\alpha1_{\{x_1=1\}}(x_1)(x_2-x_1)
\le a+\sum_{i=1}^nb_i(x_1-K_i^1)^++\sum_{j=1}^mc_j(x_2-K_j^2)^+,
\label{eqn:superhedge alpha}
\end{equation}
for all $(x_1,x_2)\in\R^2_+$. 
Let $f^*_\alpha$ denote the upper semicontinuous envelope of $f_\alpha$. We observe that \eqref{eqn:superhedge alpha} holds for $f_\alpha$ if and only if it holds also for $f^*_\alpha$. It follows that 
\begin{equation*}
\begin{split}
D_\emptyset(0)&=\inf_{0\le\alpha\le1}D_\emptyset(f_\alpha)=\inf_{0\le\alpha\le1}D_\emptyset(f^*_\alpha)=\inf_{0\le\alpha\le1} P_\emptyset(f^*_\alpha)\\
& = \inf_{0\le\alpha\le1}\alpha\E^\Q[1_{\{S_1=1\}}(S_1)(S_2-S_1)^-]=0,
\end{split}
\end{equation*}
where the third equality follows from Proposition \ref{thm:duality} and the fourth equality is due to $f^*_\alpha= \alpha1_{\{x_1=1\}}(x_1)(x_2-x_1)^-$. On the other hand, since 
\[
A^\Q_2=\sup_{\alpha\in[0,1]}\alpha\Q(S_1=1)=\frac12,
\]
we have $P_\emptyset(0)=-\E^\Q[A^\Q_2]=-\frac12$, which indicates a duality gap.

\bibliographystyle{siam}
\bibliography{refs}

\end{document}